\documentclass[letterpaper]{article} %
\usepackage[draft]{aaai25}  %
\usepackage{times}  %
\usepackage{helvet}  %
\usepackage{courier}  %
\usepackage[hyphens]{url}  %
\usepackage{graphicx} %
\usepackage{natbib}  %
\usepackage{caption} %
\usepackage{graphicx} %
\usepackage[utf8]{inputenc}
\usepackage{amsmath, amssymb, amsthm, amsfonts}
\usepackage{bm}
\usepackage{thmtools}
\usepackage{mathtools}
\usepackage{xfrac}
\usepackage[inline]{enumitem}
\usepackage{thm-restate}
\usepackage{xcolor}
\usepackage{hyperref}
\hypersetup{
    colorlinks,
    linkcolor={red!60!black},
    citecolor={blue!60!black},
    urlcolor={blue!60!black}
}
\usepackage{url}
\usepackage{multirow,array}
\usepackage{breqn}

\usepackage{tikz}
\usetikzlibrary{automata,arrows,positioning,calc}
\usetikzlibrary{shapes.misc, shapes.geometric}

\newcommand{\E}[1]{\mathbb{E}\left[ #1 \right]}
\DeclareMathOperator{\Ex}{\mathbb{E}}

\DeclareMathOperator{\Geo}{Geo}

\newcommand{\prob}[1]{\mathbb{P}\left(#1\right)}
\newcommand{\eps}{\epsilon}

\newcommand{\defeq}{\vcentcolon=}

\newcommand\simiid{\mathrel{\overset{\makebox[0pt]{\mbox{\normalfont\tiny iid}}}{\sim}}}
\newcommand{\norm}[1]{\left\lVert#1\right\rVert}

\let\originalmiddle=\middle
\def\middle#1{\mathrel{}\originalmiddle#1\mathrel{}}

\theoremstyle{plain}
\newtheorem{thm}{Theorem}
\newtheorem{lemma}{Lemma}

\newtheorem{conj}{Conjecture}
\newtheorem{cor}{Corollary}
\newtheorem{defn}{Definition}
\theoremstyle{definition}

\theoremstyle{remark}

\newtheorem{eg}{Example}

\usepackage[capitalise]{cleveref}
\Crefname{thm}{Theorem}{Theorems}
\Crefname{prop}{Proposition}{Propositions}
\Crefname{rk}{Remark}{Remarks}
\Crefname{eg}{Example}{Examples}
\Crefname{asm}{Assumption}{Assumptions}
\Crefname{defn}{Definition}{Definitions}
\Crefname{cor}{Corollary}{Corollaries}

\usepackage{ifthen}
\newboolean{commentsactivated}
\setboolean{commentsactivated}{true}
\newcommand{\co}[1]{\ifthenelse{\boolean{commentsactivated}}{{\color{red} {\em CO: #1 }}}{}}
\newcommand{\ec}[1]{\ifthenelse{\boolean{commentsactivated}}{{\color{teal} {\em EC: #1 }}}{}}
\newcommand{\vc}[1]{\ifthenelse{\boolean{commentsactivated}}{{\color{olive} {\em VC: #1 }}}{}}

\newcommand{\note}[1]{\ifthenelse{\boolean{commentsactivated}}{{\color{magenta} {#1}}}{}}

\usepackage{csquotes}
\usepackage[utf8]{inputenc}
\usepackage{algorithm}
\usepackage[noend]{algorithmic}

\newcommand{\p}{\bm{p}}

\newcommand{\apply}{\mathit{apply}}
\newcommand{\consistent}{\mathit{consistent}}
\usepackage{mathtools}

\newcommand{\eGpB}{$\epsilon$Grounded$\pi$Bot }

\title{Characterising Simulation-Based Program Equilibria}
\author{Emery Cooper, Caspar Oesterheld, Vincent Conitzer}
\affiliations{Carnegie Mellon University\\
emerycooper@cmu.edu, oesterheld@cmu.edu, conitzer@cs.cmu.edu}
\date{\today}

\begin{document}

\maketitle
\begin{abstract}
In \citeauthor{tennenholtz}'s program equilibrium, players of a game submit programs to play on their behalf. Each program receives the other programs' source code and outputs an action. This can model interactions involving AI agents, mutually transparent institutions, or commitments. \citet{tennenholtz} proves a folk theorem for program games, but the equilibria constructed are very brittle.
We therefore consider \emph{simulation-based programs} -- i.e., programs that work by running opponents' programs. These are relatively robust (in particular, two programs that act the same are treated the same) and are more practical than proof-based approaches.
\citeauthor{oesterheld2019robust}'s \citeyearpar{oesterheld2019robust} $\eps$Grounded$\pi$Bot is such an approach. Unfortunately, it is not generally applicable to games of three or more players, and only allows for a limited range of equilibria in two player games.
In this paper, we propose a generalisation to \citeauthor{oesterheld2019robust}'s \citeyearpar{oesterheld2019robust} $\eps$Grounded$\pi$Bot.
We prove a folk theorem for our programs in a setting with access to a shared source of randomness. We then characterise their equilibria in a setting without shared randomness. Both with and without shared randomness, we achieve a much wider range of equilibria than \citeauthor{oesterheld2019robust}'s \citeyearpar{oesterheld2019robust} $\eps$Grounded$\pi$Bot. Finally, we explore the limits of simulation-based program equilibrium, showing that the \citeauthor{tennenholtz} folk theorem cannot be attained by simulation-based programs without access to shared randomness.

\end{abstract}

\section{Introduction}

Consider a game in which, rather than choosing actions, players submit programs (\citealt{McAfee1984}; \citealt{Howard1988}; \citealt{Rubinstein1998}, Sect.\ 10.4; \citealt{tennenholtz}). Programs receive as input the source code of the opponents' programs, and output actions. E.g., in the (one-shot) prisoner's dilemma, the programs might be CliqueBots: `cooperate if the other program is equal to this program, else defect'. CliqueBots form a cooperative equilibrium. This illustrates how such games -- \emph{program games} -- allow for new equilibria.

Traditional game-theoretic concepts, such as Nash equilibrium, have been developed without mutual transparency in mind, and may therefore be less relevant to interactions involving AI agents whose source codes might be available to other agents \citep{ConitzerCoopAI}.
Program games could model interactions between principals designing AI agents, interactions between AI agents themselves, or smart contracts. This could include interactions between large language models where the programs are prompts provided by humans and interpreted by the models.
Program games can also model some interactions between humans, such as those involving mutually conditional commitments, or transparent institutions \citep{Critch2022}.

\citet{tennenholtz} proves a folk theorem for program equilibria \citep[cf.][Sect.\ 10.4]{Rubinstein1998}: any set of feasible and individually rational (i.e., better-than-minimax) payoffs of a game is achievable in a program equilibrium. However, this result is proved using programs based on syntactic comparison, such as the above CliqueBot in the prisoner's dilemma. Such equilibria are extremely fragile: if someone, perhaps from a different community, writes a slightly different version of CliqueBot (CliqueBot'), then CliqueBot and CliqueBot' will defect on each other.

Previous authors have proposed different ways of achieving more robust program equilibria. One proposal is that of Löbian programs \citep{Barasz2014,Critch2019,Critch2022,oesterheld2022notecompatibility}. In the prisoner's dilemma, the Löbian FairBot roughly works by cooperating if and only if it can prove the opponent cooperates. The Löbian FairBots cooperate against each other. Unfortunately, there are practical difficulties with this approach. If the opponent program is not very interpretable (e.g., a large neural net), it may be difficult or impossible to prove things about its output.

We therefore turn our attention to simulation-based approaches, such as \citeauthor{oesterheld2019robust}'s \citeyearpar{oesterheld2019robust} $\eps$GroundedFairBot and $\eps$Grounded$\pi$Bot. In the prisoner's dilemma, the $\eps$GroundedFairBot cooperates with probability $\eps$ and otherwise simulates the opponent and cooperates if and only if in the simulation the other player cooperates. If two $\eps$GroundedFairBots play against each other, the simulated player may in turn simulate the other, etc., but eventually the $\eps$-probability case will be triggered, ensuring termination.
This program is easy to implement and can be used against programs that are not very interpretable.

We can view \citeauthor{oesterheld2019robust}'s (\citeyear{oesterheld2019robust}) $\eps$Grounded$\pi$Bot as simulating a repeated game. With probability $\epsilon$, the program is on the first `time step', and does not observe any previous actions. Otherwise, it simulates the other program. If the other program is another $\eps$Grounded$\pi$Bot, that program will again be on the first time step with probability $\eps$, and will do another simulation with probability $1-\eps$. Overall, the total number of simulations (or `time steps') is then geometric, with programs at each time step observing only the immediate previous time step. \citet{oesterheld2019robust} proves a correspondence between program games of $\eps$Grounded$\pi$Bots and repeated games. Since only the previous time step is observed, the policy $\pi$ must be \emph{myopic} -- i.e., only depend on the last action of the opponent.\footnote{This is similar to the Markov property, viewing the actions taken on the previous time step as the state -- given the last action of the opponent, the $\eps$Grounded$\pi$Bot's action does not depend on earlier actions.}

We focus on the following question: What equilibria can be obtained by simulation-based approaches? Specifically, can they obtain all equilibria in \citeauthor{tennenholtz}' folk theorem?

Unfortunately, there is a huge gap between the \citet{oesterheld2019robust} results and \citeauthor{tennenholtz}'s folk theorem. The $\eps$Grounded$\pi$Bot cannot perform more than one simulation (when its opponents are also $\eps$Grounded$\pi$Bots) without running into halting issues due to infinite recursion, and so does not work well for games of more than two players. (Simulating two other players with probability $1-\eps$ each leads to each simulation in expectation generating $2-2\eps$ new ones, resulting in a hydra of simulations -- cf.\ \Cref{app:2019_halting}.) Even in two-player games, the equilibria it can enforce fall short of the folk theorem.

\begin{eg}\label{eg:intro_game}
    Consider the following three-player game (payoff matrix shown in \Cref{tab:eg_1}):%
    \begin{itemize}
        \item The utilities are \emph{additively separable}, i.e., for each player $i$, that player's utility $u_i$ may be decomposed into a sum of functions $u_{ij}$, where $u_{ij}$ depends only on the action of player $j$: $u_i(a_1,a_2,a_3) = u_{i1}(a_1)+u_{i2}(a_2)+u_{i3}(a_3)$.
        \item Players 2 and 3 may either Cooperate ($C$), granting all players an additional 3 utility, or Defect ($D$), granting 8 utility to the defecting player. Formally, for $j=2,3$ we have $u_{ij}(C)=3$ for all $i$, while $u_{ij}(D) = 8$ if $i=j$ and 0 otherwise.
        \item Player 1 may either Cooperate ($C$), Punish player 2 ($P_2$), or Punish player 3 ($P_3$). Playing $P_2$ (resp. $P_3$) takes 3 utility from player 2 (resp. 3) and gives it to player 3 (resp. 2). Formally, $u_{i1}(C) = 0$ for all $i$, while $u_{21}(P_2) = u_{31}(P_3) = -3$, $u_{21}(P_3) = u_{31}(P_2) = 3$, and $u_{11}(P_2)=u_{11}(P_3) = 0$.
    \end{itemize}
    \end{eg}
    \begin{table}[tbh]
\small
\centering
\setlength{\extrarowheight}{0pt}
\begin{tabular}{cc|c|c|}
  \multicolumn{1}{c}{} & \multicolumn{1}{c}{} & \multicolumn{2}{c}{Player 2}\\
  \multicolumn{1}{c}{Plr 3} & \multicolumn{1}{c}{Plr 1} & \multicolumn{1}{c}{$C$}  & \multicolumn{1}{c}{$D$}\\
\cline{3-4}
 & $C$ & $6,6,6$ & $3,11,3$\\
$C$ & $P_2$ & $6,3,9$ & $3,8,6$\\
 & $P_3$ & $6,9,3$  & $3,14,0$\\
\cline{3-4}
 & $C$ & $3,3,11$ & $0,8,8$\\
$D$ & $P_2$ & $3,0,14$ & $0,5,11$\\
 & $P_3$ & $3,6,8$ & $0,11,5$ \\
\cline{3-4}
\end{tabular}
\caption{The payoff matrix for \Cref{eg:intro_game}}
\label{tab:eg_1}
\end{table}
    Can we obtain a $(C,C,C)$ equilibrium in this game with simulation-based programs?
    With the \citet{oesterheld2019robust} $\eps$Grounded$\pi$Bot, the answer is no. As noted above, it cannot be na\"ively extended to simulate more than one other $\eps$Grounded$\pi$Bot while still halting with probability 1 (see \Cref{app:2019_halting} for details). In some three-player games, the $\eps$Grounded$\pi$Bots can work around this, for example, by choosing one player to simulate at random.  But this will not work in Example~\ref{eg:intro_game}. Intuitively, player 1 needs to be able to punish a defecting player; but if player 1 randomises which player to simulate, and then sees that player (say, 2) defect in simulation, then player 1 will not know which player defected `first' in the mental repeated game -- e.g., it could be that 2 is defecting only because 3 defected in a \emph{subsimulation}. That is, the simulated version of player 2 might in turn simulate player 3, observe player 3 defecting in that simulation, and defect so as to penalise player 3 for this. Thus, player 1 will not know which player to punish. See \Cref{app:2019_3player_schemes} for details.

    \paragraph{Contributions} 
Now, imagine the program worked as follows instead: First it randomly determines its own `time step' in the mental repeated game, according to a geometric distribution. Then if its time step is 1, it immediately outputs an action. Otherwise, it simulates the other programs, but manipulates those simulations' sources of randomness to ensure that (assuming the other programs operate similarly) they believe themselves to be on earlier time steps. This would then allow it to simulate multiple different opponents (and different time steps) while still halting. Our programs will be based on this idea.

    We first consider a setting (\Cref{sec:correl}) where programs have access to a shared random sequence. We propose a generalisation of the \citeauthor{oesterheld2019robust} (\citeyear{oesterheld2019robust}) $\eps$Grounded$\pi$Bot for this setting, which we call a \textit{correlated $\eps$Grounded$\pi$Bot}. This correlated $\eps$Grounded$\pi$Bot may be viewed as simulating a repeated game between $n$ players, and responding to the observed action history. We relate payoffs in games where all players but (at most) one are correlated $\eps$Grounded$\pi$Bots to payoffs in repeated games. We hence obtain a folk theorem: any feasible and strictly individually rational payoffs can be obtained in a program equilibrium of correlated $\eps$Grounded$\pi$Bots. From this, we can immediately see that these programs can solve \Cref{eg:intro_game}, since the profile $(C,C,C)$ is better for each player than that player's minimax utility. Our correlated $\eps$Grounded$\pi$Bots can be seen as a (decentralised) implementation of \citeauthor{kovarik2024recursive}'s joint simulation device (cf.\ \Cref{sec:related_work}).

    We then consider the same approach in a setting without access to shared randomness (\Cref{sec:uncorrel}). We call the resulting program an `uncorrelated $\eps$Grounded$\pi$Bot'. We characterise the equilibria that can be attained by our uncorrelated $\eps$Grounded$\pi$Bots. We give examples to show that our uncorrelated $\eps$Grounded$\pi$Bots attain equilibria that the \citet{oesterheld2019robust} $\eps$Grounded$\pi$Bots cannot, even with $2$ players. Nonetheless, our uncorrelated $\eps$Grounded$\pi$Bots do not achieve the full program equilibrium folk theorem, as we show with an example. While there is in general no simple relationship between program games of uncorrelated $\eps$Grounded$\pi$Bots and repeated games, we do obtain such a result in the special case of games with additively separable utilities. From this, we derive a folk theorem in the case of games with additively separable utilities. Hence, the uncorrelated $\eps$Grounded$\pi$Bots can still achieve a $(C,C,C)$ equilibrium in \Cref{eg:intro_game}.

    Finally, we ask whether more general simulation-based programs allow for additional equilibria in the setting without shared randomness (\Cref{sec:simulationist}). We give an example to show that it is possible to attain equilibria that cannot be attained by the uncorrelated $\eps$Grounded$\pi$Bot. We then prove a negative result on the equilibria attainable by simulation-based programs. From this, we show that without shared randomness, simulation-based programs cannot give us the full \citet{tennenholtz} folk theorem. It remains an open question exactly which equilibria can be attained.

\section{Background and Preliminaries}
    We will sometimes denote the tuple $(x_1,\ldots, x_n)$ by $x_{1:n}$, and the tuple $(x_1,\ldots, x_{i-1},x_{i+1},\ldots, x_n)$ by $x_{-i}$. Similarly, we will use `players $-i$' to refer to all players of a game but player $i$.
\paragraph{Normal-form games}
We assume familiarity with game theory. See, e.g., \citet{osborne2004} for an introduction.
We introduce some definitions for normal-form games.

    An $n$-player \emph{normal-form} game comprises finite \emph{action sets} $\mathcal{A}_i$ and \emph{utility functions} $u_i: \prod_i \mathcal{A}_i \rightarrow \mathbb{R}$ for each player $i \in \{1,\ldots n\}$.
    A \emph{strategy} for player $i$ is a probability distribution $s_i\in \Delta(\mathcal{A}_i)$ over actions for player $i$. A \emph{strategy profile} is a tuple $(s_1,\ldots s_n)$ of strategies for all players.

    Given a joint distribution over \emph{outcomes} $(a_1,\ldots, a_n)$, the \emph{expected utility} for player $i$ is, writing $A_i$ for the $i$th player's action (a random variable), $
        \E{u_i} = \E{u_i(A_{1:n})} = \sum_{a_{1:n}} \prob{A_{1:n} = a_{1:n}}u_i(a_1,\ldots, a_n)$.
    When players' actions are independent, with player $i$ following strategy $s_i$ for each $i$, write $u_i(s_1,\ldots, s_n) \defeq \E{u_i(A_{1:n})} = \sum_{a_{1:n}}u_i(a_{1:n})\prod_j\prob{A_j = a_j}$.

        We call a set of expected payoffs $v_{1:n}$ \emph{individually rational (without coordination)}
        if it is at least as good for each player as that player's minimax utility, i.e., $v_i \geq \min_{s_{-i} \in \prod_{j\neq i}\Delta(\mathcal{A}_j)}\allowbreak\max_{a_i\in \mathcal{A}_i} u_i(s_{-i},a_i)$ for all $i$. Say $v_{1:n}$ is \emph{strictly} individually rational if this inequality is strict for all $i$.

            Say that a set of payoffs $v_{1:n}$ is \emph{feasible (without correlation)} if it is achieved by some set of (independently sampled) strategies, i.e., if for some $s_{1:n}$, for all $i$, $v_i = u_i(s_{1:n})$.
    Meanwhile, $v_{1:n}$ is \emph{feasible with correlation} if it is achieved in expectation by some set of strategies not sampled independently: i.e., if when $A_{1:n}$ follows some joint distribution over $\prod_i \mathcal{A}_i$, we have $v_i = \E{u_i(A_{1:n})}$.

\paragraph{Program games}
We now define \emph{program equilibrium} \citep{tennenholtz}. We will consider programs that are themselves deterministic, but take as input one or more random sequences. We also assume that programs always have access to their own source code. Write $\apply$ for the function that takes as input a program and a tuple of input arguments and runs the program on the input arguments, outputting whatever the program outputs (if anything).

\begin{defn}[Program games]
Given an $n$-player normal-form game $G$, the \emph{correlated program game} of $G$ is a game where player $i$'s action set is the set of programs taking as input a list of $n-1$ opponent programs, and sequences $(r_m)_{m\geq 0}\in[0,1]^\infty$ and $(x^i_m)_{m\geq 0}\in [0,1]^\infty$, and outputting an action in $\mathcal{A}_i$. Given a \emph{program profile} $(p_1,\ldots, p_n)$ of programs that almost surely halt against each other (i.e., halt against each other with probability 1), the \emph{outcome distribution} is given by the distribution of output actions, $A_{1:n}$, where $A_i = \apply(p_i,(p_{-i},(r_m),(x_m^i)))$, and $(r_m)$, $(x_m^1),\ldots (x_m^n)$ are independent i.i.d.\ $U[0,1]$ sequences. The \emph{expected utility} for player $i$ is then $\E{u_i(A_{1:n})}$.

An \emph{uncorrelated program game} is the same, except that for each $i$, program $i$ only receives a single sequence $(r^i_m)$, independent for each $i$.
\end{defn}

If the programs do not halt against each other, we leave the outcome distribution and expected utility undefined. We will show that the programs we define halt against each other and a range of other programs. When we later consider more general programs, we will place restrictions on them to prevent them from strategically not halting (which our earlier programs also satisfy). While we can consider players randomising over programs (cf.\ \Cref{sec:princip_rand}), it will generally suffice to imagine each player submitting a single program.

    Given an $n$-player (correlated or uncorrelated) program game and programs $p_{-i}$ for players $-i$, a program $p_i$ is a \emph{best response} to $p_{-i}$ if it achieves maximal expected utility against $p_{-i}$ (among programs $p_i'$ such that $p_{-i},p_i'$ all almost surely halt against one another). Say $(p_1,\ldots p_n)$ is a \emph{program equilibrium} if for all $i$, $p_i$ is a best response to $p_{-i}$.

\citet{tennenholtz} proves the following folk theorem for program games (where the set of programs for each player is restricted to avoid issues related to halting). We will be interested in whether this result (or something similar) holds for simulation-based programs.
\begin{thm}
    Payoffs $v_{1:n}$ are achievable in an equilibrium of the uncorrelated program game of $G$ (given aforementioned restrictions on the programs available to each player) if and only if they are feasible (without correlation) and individually rational.
\end{thm}

\section{Correlated Program Games}\label{sec:correl}
We first consider correlated program games.
We use an additional variant on $\apply$ calls for this setting:
    Define $\apply^*(p_i,(p_{-i},(r_m),(x_m)))$ to be equal to $\apply(p_i,(p_{-i},(r_m),(x_m)))$ if the variables $(x_m)$ are not \emph{accessed} except inside further $\apply^*$ calls, and $R_i$, a special symbol to indicate that the private random sequence was accessed, otherwise. For example, in the random-access machine model this would trigger whenever the simulated program runs the command for reading (with either direct or indirect addressing) one of memory addresses containing the independent random bits (e.g., loading from, comparing with the contents of, or performing an arithmetic operation using the contents of such a memory location, in the formalism of `MIX' from \citet[Sec.\ 1.3.1, Volume 1]{knuth97}). We do not include simply passing the $(x_m)$ on via $\apply$ calls as `accessing' them -- we may assume that these pass on only a pointer to the $(x_m)$.

The idea of $\apply^*$ calls is to screen off any dependence on private random variables (by replacing the output in cases where it could depend on them). Computing whether a program's output \emph{depends} on the $(x_m)$ is undecidable (by Rice's theorem), so we instead look at whether the program \emph{accesses} the $(x_m)$. We will discuss the purpose of this later.

\paragraph{Our algorithm}

We present our algorithm, which we call a \emph{correlated $\eps$Grounded$\pi$Bot}, in \Cref{algorithm:pi-bot-correl}. Note that the $\eps$Grounded$\pi$Bot \emph{policy}, $\pi_i$, takes as input action histories in which some actions may be replaced with $R_j$s, and outputs an action for player $i$ (possibly also depending on random variables $r_0$ and $x_0$). The program essentially simulates a repeated game between the programs. It works as follows (see also \Cref{fig:simulations} for an illustration):
\begin{itemize}
    \item First, it determines the number $T$ of previous `time steps' as the number of elements of the shared random sequence $(r_m)$ until the first element that is less than $\eps$. This gives it a geometric distribution over its time step. %
    \item If $T=0$, it immediately outputs its policy $\pi_i$ applied to the empty history.
    \item Otherwise, it simulates all programs (including itself) at earlier time steps. It does this by calling them with (between 1 and $T$) initial elements of the shared random sequence removed -- effectively putting them on earlier `time steps' (or with smaller values of $T$). This ensures that $\eps$Grounded$\pi_i$ bots will halt against one another, since the tree of $\apply$ calls will then have depth at most $T$.
    \item The simulated programs are not able to tell how many `time steps' there are `after' them, and hence cannot tell that they are simulated. This is analogous to how in a repeated game that ends with probability $\eps$ each round, players do not know if they are on the last time step.
    \item It then outputs its policy $\pi_i$ applied to the action history.
\end{itemize}

\begin{figure}[h]
    \centering
	\begin{tikzpicture}[->, >=stealth', auto, semithick, on grid,
	every state/.style = {fill=none,draw=black,thick,text=black},
	every text node part/.style={align=center}]
	\node[state, rectangle]    (0) {$p_1$ input $(r_m)_{m\geq 0}, (x_m)_{m\geq 0}$};
    \node[state, rectangle, draw = none] (label)[above left = 0.5cm and 3.3cm of 0] {$T=3$};
	\node[state, rounded rectangle] (1)[below = 1.5cm of 0] {$p_{1}, \ldots, p_n$ input $(r_m)_{m\geq 1}, (x_m)_{m\geq 0}$};
    \node[state, rounded rectangle] (2)[below = 1.5cm of 1] {$p_{1}, \ldots, p_n$ input $(r_m)_{m\geq 2}, (x_m)_{m\geq 0}$};
    \node[state, rounded rectangle] (3)[below = 1.5cm of 2] {$p_{1}, \ldots, p_n$ input $(r_m)_{m\geq 3}, (x_m)_{m\geq 0}$};
	\path
	(0) edge	(1)
    (1) edge (2)
    (2) edge (3)
    (0) edge[bend right, out = 280, in = 260] (2)
    (0) edge[bend right, out = 270, in = 260] (3)
    (1) edge[bend left, out = 80, in = 100] (3);
	\end{tikzpicture}
	\caption{The simulation tree in a game between $n$ correlated $\eps$Grounded$\pi$Bots with common $\eps$ in the case $T= 3$ (i.e., when $r_0,r_1,r_2\geq \eps$ and $r_3<\eps$). An arrow from a group of programs with one input, to a second group of programs with a second input, indicates that the first group with the first input all run all of the second group of programs with the second input (in this case, via an $\apply^*$ call). We omit the part of the input that is simply the list of other programs.}\label{fig:simulations}
\end{figure}
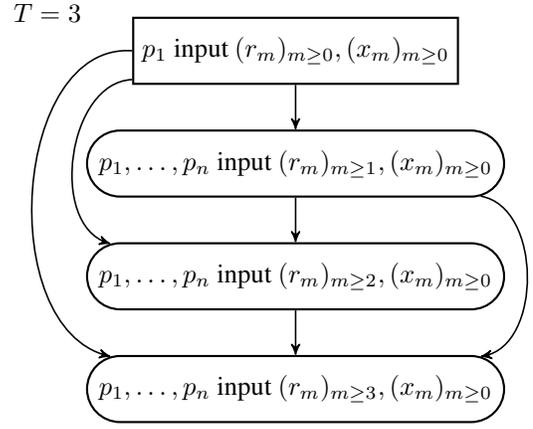

Note that aside from the availability of shared random variables, the main difference between our algorithm and the $\eps$Grounded$\pi$Bot from \citet{oesterheld2019robust} is that we allow for more than 2 players, and look back at more than one time step. This is possible because having the random variables as an input argument allows for manipulating the random variables that simulated programs receive, thus allowing us to correlate their maximum simulation depth. This in turn allows us to consider non-myopic policies $\pi_i$. One final thing to note is the use of $\apply^*$ calls rather than $\apply$ calls. This serves a few purposes. First, programs independently randomising can cause problems, as we will see in later sections, and using $\apply^*$ calls allows for punishing programs for doing this. Whilst we could additionally allow our programs to observe the results of the randomisation using the $(x_m)$s, this would result in the simulations being correlated in a way that the actual programs would not be, which would complicate our results. Finally, by making it so the output of the $\eps$Grounded$\pi$Bots themselves will not depend on the private random bits, it allows the history to still be observed normally after such randomisation occurs, allowing for non-grim-trigger strategies.

\begin{algorithm}[tb]
\caption{Correlated $\eps$Grounded$\pi_i$Bot}
\label{algorithm:pi-bot-correl}
\begin{algorithmic}
    \STATE {\bfseries Input:} Programs $p_{-i}$, sequences $(r_m)_{m=0}^\infty, (x_m)_{m=0}^\infty\in[0,1]^{\infty}$
    \STATE $T \leftarrow \min \{t:r_t <\epsilon \}$
    \FOR{$t=1$ to $T$}
        \FOR{$j=1$ to $n$}
        \STATE $a_j^t \leftarrow \apply^*(p_j, (p_{-j},(r_m)_{m\geq T+1-t},(x_m)_{m\geq 0}))$
    \ENDFOR
    \ENDFOR
    \RETURN{$\pi_i(\bm{a}^{1:T}; r_0,x_0)$}
\end{algorithmic}
\end{algorithm}

\paragraph{Halting results}

We show that the $\eps$Grounded$\pi$Bots halt against each other, and that this halting is relatively robust. Note that checking whether $r_t<\eps$ might take arbitrarily long (or be impossible) -- we abstract away this problem by simply assuming this computation is instant. Similarly, we assume throughout that the $\eps$Grounded$\pi$Bot policies $\pi_i$ are computable.

 \begin{restatable}{thm}{correlhalting}\label{thm:correl_halting}
Suppose that $p_i$ is a correlated $\eps_i$Grounded$\pi_i$Bot with $\eps_i> 0$, for $i=1,\ldots n$. Then programs $p_1,\ldots, p_n$ almost surely halt against each other.
\end{restatable}
\begin{proof}
    See \Cref{pf:correlhalting}, p.\ \pageref{pf:correlhalting}.
\end{proof}

The intuition for this result is as follows: We can consider the `time step' ($T_i$) from the perspective of the program with minimal $\eps_i$. The $\eps$Grounded$\pi$Bots' simulations are always on a smaller time step (in the sense of $T_i$) than they are. On time step 1 (if not earlier), they must halt without running any simulations. Thus, the overall process must halt.

If one of the programs is not an $\eps$Grounded$\pi$Bot, the programs still halt under certain conditions.
    Given program $p_i$ with input programs $p_{-i}$, say that an $\apply$ (resp. $\apply^*$ call) is \emph{standard} if the program being called is $p_j$ for some $j$, with input programs $p_{-j}$.
    
\begin{restatable}{thm}{correlHaltingExt}\label{thm:correl_halting_ext}
Suppose that program $p_i$ has the following properties:
\begin{enumerate*}[label = (\roman*)]
    \item It never makes a standard $\apply/\apply^*$ call to itself.
    \item Its standard $\apply/\apply^*$ calls all have first input sequence of the form $(r_m)_{m\geq k}$ with $0\leq k \leq c$ for some constant $c$.
    \item $p_i$ almost surely makes only finitely many standard $\apply/\apply^*$ calls directly (i.e., not counting those made inside other $\apply$ calls).
    \item $p_i$ almost surely halts whenever all its standard $\apply/\apply^*$ calls halt.
    \item For all events $A$ with $\prob{((r_m),(x^i_m))\in A} = 0$, $p_i$ almost surely makes no standard $\apply$/$\apply^*$ calls with input sequences in $A$.
\end{enumerate*}
Then if for each $j\neq i$, $p_j$ is a (correlated) $\eps_j$Grounded$\pi_j$Bot with $\eps_j>0$, programs $p_{1:n}$ almost surely halt against each other.
 \end{restatable}
 \begin{proof}
    See \Cref{pf:correlHaltingExt}, p.\ \pageref{pf:correlHaltingExt}.
\end{proof}

\subsubsection{Expected runtime}
By default, $\eps$Grounded$\pi$Bots have infinite expected runtime for small $\eps$, despite halting with probability 1. This is due to an exponential blow-up of the number of simulations, with the number of simulations increasing by a factor of at least $n$ at each level of simulation. However, these are mostly duplicates of each other -- there are only $n$ distinct simulations at each `time step' (i.e, with each possible pair of input sequences).
Thus, we may use \emph{memoization} \citep[][Exercise 3.27]{michie1968memo,norvig1991techniques,abelson1996structure} to circumvent this problem -- if we store the results of computations of $\apply^*$ calls, we can avoid redoing them, and obtain finite expected runtime.

\begin{restatable}{thm}{expectedRuntimeCorrel}\label{thm:expected_runtime_correl}
Suppose that program $p_i$ is a correlated $\eps_i$Grounded$\pi_i$Bot for $i = 1,\ldots n$. Then if at least two $p_i$ have $\eps_i<1/2$, every $p_i$ with $\eps_i < 1$ has infinite expected runtime against $p_{-i}$. However, with memoization, the expected runtime is finite and polynomial in $\frac{1}{\min_i \eps_i}$, provided that for all $i$, $\pi_i(\bm{a}^{1:t})$ runs in polynomial time in $t$.
\end{restatable}
\begin{proof}
    See \Cref{pf:expectedRuntimeCorrel}, p.\ \pageref{pf:expectedRuntimeCorrel}.
\end{proof}

\subsubsection{Repeated games with correlating signals and screening}

Our results will come from analogy to repeated games. We will first need to define a slight variant on repeated games. These will have two new features. \emph{Screening}: we allow each player at each time step to `screen' their action, thereby hiding the action from the other players. This is analogous to actions being `screened' from the $\eps$Grounded$\pi$Bots by the $\apply^*$ calls when they are based on private randomness. \emph{Correlating signals:} We assume that at each time step, players all receive the same $U[0,1]$ random variable.
We now define screened histories and policies. %

    Given a game $G$ with $n$ players, define a \emph{screened history} of length $t$ to be an element $\bm{\alpha}^{1:t} \in (\prod_i (\{R_i\}\cup \mathcal{A}_i))^t$ (where $\alpha_i^k$ then denotes the action of player $i$ on time step $k$). Say that a screened history is \emph{compatible} with an \emph{unscreened history} in $(\prod_i \cup \mathcal{A}_i)^t$ if they agree up to screening (i.e., they agree on all elements that are not $R_i$ (for some $i$) in the screened history). That is, screened history $\bm{\alpha}^{1:t} \in (\prod_i (\{R_i\}\cup \mathcal{A}_i))^t$ is compatible with unscreened history $\bm{\beta}^{1:t} \in (\prod_i (\mathcal{A}_i))^t$ iff for all $j$, $t'$, whenever $\alpha_j^{t'} \neq R_j$ we have $\alpha_j^{t'}=\beta_j^{t'}$.
    
    Similarly, a \emph{screened policy} for player $j$ is a map from, for $t=0,1,
    \ldots$, screened histories of length $t$ and correlating signals $q \in [0,1]$ to $\Delta(\mathcal{A}_j \cup \{R_j\})$. An \emph{unscreened policy} in a repeated game with correlation and screening is the same, but instead mapping to $\Delta(\mathcal{A}_j)$. Say that a screened policy  $\pi_i^*$ is \emph{compatible} with an unscreened policy $\pi_i$ if they agree up to screening -- i.e., for each action $a_j \in \mathcal{A}_j$, each history $\bm{\alpha}^{1:t}$, and each $q$, $\pi_i^*$ assigns no more probability than does $\pi_i$ to action $a_j$. That is,
    $\pi_i^*(a_j\mid \bm{\alpha}^{1:t};q) \leq \pi_i(a_j\mid \bm{\alpha}^{1:t};q)$, writing $\pi_i(a_j\mid \bm{\alpha}^{1:t};q)$ for the probability mass given by $\pi_i(\bm{\alpha}^{1:t};q)$ to $a_j$.

    Then, given screened and unscreened policies $\pi_i^*$ and $\pi_i$ for player $i$, define their joint distribution by $(\tau_i,\tau_i^*)((a_j,a_j')\mid \bm{\alpha}^{1:t};q) \defeq \tau_i^*(a_j\mid \bm{\alpha}^{1:t};q)$ if $a'_j=a_j$, and $\tau_i(a_j\mid \bm{\alpha}^{1:t};q)-\tau_i^*(a_j\mid \bm{\alpha}^{1:t};q)$ if $a'_j=R_j$, and 0 otherwise. Given a repeated game with correlation and screening, and screened and unscreened policy profiles $\bm{\tau}^*$ and $\bm{\tau}$ respectively, the probability of the screened and unscreened histories at time $t$ (given that the game does not end before then) being equal to (compatible) histories $\bm{\alpha}^{1:t}$ and $\bm{\beta}^{1:t}$ is, where $q_1,q_2,\ldots \simiid U[0,1]$, defined as 
    $\prod_{k = 1}^t\prob{\bm{\tau}^*(\bm{\alpha}^{1:k-1};q_k) = \bm{\alpha}^{k}\cap \bm{\tau}(\bm{\alpha}^{1:k-1};q_k) = \bm{\beta}^{k}} \defeq \prod_{k = 1}^t\E{\prod_{i=1}^n (\tau_i,\tau_i^*)((\beta_i^k,\alpha_i^k)\mid\bm{\alpha}^{1:k-1};q_k)}$. (Thus, we assume that given correlating signals, samples from policies are independent.)

    The history length is independent of the actions taken, with, for \emph{discount factor} $1-\eps$, probability $(1-\eps)^{t-1}\eps$ of the final history having length $t$. This therefore defines a distribution over the final histories, and hence over the final unscreened history.
    
    Given that the game ends after action history $\bm{a}^{1:t}$, player $i$ receives utility $\sum_{t'=1}^t u_i(a_1^{t'},\ldots a_n^{t'})$. The \emph{expected utility} for player $i$ is then the expectation of the utility for player $i$, taken with respect to the distribution over unscreened histories.

    Screened and unscreened policy profiles $\bm{\tau^*}, \bm{\tau}$ are an \emph{equilibrium} if for all $i$, $(\tau_i^*,\tau_i)$ maximizes expected utility against $(\tau_{-i}^*,\tau_{-i})$ among pairs of compatible screened and unscreened policies for player $i$.

Given a game $G$, we write $G_\eps$ for the corresponding correlated repeated game with screening with discount factor $1-\eps$. We obtain the following folk theorem, as for usual repeated games. The idea is to consider policies that punish screening, such that it won't happen on the path of play.

\begin{restatable}{thm}{repFolkThm}\label{thm:repFolkThm}
    Suppose that payoffs $v_{1:n}$ for game $G$ are feasible with correlation and strictly individually rational. Then for $\eps$ sufficiently small, there exists an equilibrium of $G_\eps$ in which player $i$ receives expected utility $\frac{v_i}{\eps}$, for all $i$.
\end{restatable}
\begin{proof}
    See \Cref{pf:repFolkThm}, p.\ \pageref{pf:repFolkThm}.
\end{proof}

\subsubsection{Relating the program game to the repeated game}

We now translate from programs to policies in correlated repeated games with screening. Essentially, the policy of a program at time $t$ will be determined by its action distribution conditional on $T+1=t$, given the simulated actions of opponents at earlier time steps. The screening will be determined by whether the program accesses its private random sequence.

    Given programs $p_1,\ldots p_n$ for game $G$, we define the \emph{associated screened and unscreened policies} for the repeated game $G_\eps$ as follows. The screened policy is, where we use $\omega$ to denote the empty history: $\tau_i^*(\omega;q) \defeq \Ex[\apply^*(p_i,(p_{-i},(r_m),(x_m)))\mid r_0=\eps q]$ and $\tau_i^*(\bm{\alpha}^{1:t};q) \defeq \Ex[\apply^*(p_i,(p_{-i},(r_m),(x_m)))\mid r_0 = \eps + (1-\eps)q,\ \bm{a}^{1:t}=\bm{\alpha}^{1:t},T=t]$, where $T$ and $\bm{a}^{1:t}$ are as defined in Algorithm \ref{algorithm:pi-bot-correl}, i.e., $T$ is minimal such that $r_T<\eps$ and $a_j^k \allowbreak\defeq \apply^*(p_j,(p_{-j},\allowbreak(r_m)_{m\geq T+1 -k},\allowbreak(x_m)_{m\geq 0}))$.
        The unscreened policy is the same but with $\apply$ instead of $\apply^*$ calls: $\tau_i(\omega;q) \defeq \E{\apply(p_i,(p_{-i},(r_m),(x_m)))\mid r_0=\eps q}$ and $\tau_i(\bm{\alpha}^{1:t};q) \defeq \Ex[\apply(p_i,(p_{-i},(r_m),(x_m)))\mid r_0 =\eps + (1-\eps)q,\ \bm{a}^{1:t}=\bm{\alpha}^{1:t},T=t].$ If for some $\bm{\alpha}^{1:t}$ the event $
        {\bm{a}^{1:t}=\bm{\alpha}^{1:t}
        }$ occurs with probability 0, we simply define the screened and unscreened policies for $G_\eps$ arbitrarily at $\bm{\alpha}^{1:t}$. (For our results, it does not matter how the screened and unscreened policies are defined at histories with probability 0, since these histories will also have probability 0 in the associated repeated game with the associated policies.)
Note that these form compatible screened and unscreened policies, where they are defined.

For $\eps$Grounded$\pi_i$Bots, the associated repeated game policy then corresponds to the policy $\pi_i$:
\begin{restatable}{lemma}{eGpBPolicyCorrel}\label{lemma:eGpB_policy_correl}
    If $p_j$ is a correlated $\eps$Grounded$\pi_j$Bot, then regardless of programs $p_{-j}$, its associated repeated game policies are, where defined, given by screened policy: $\tau_j^*(\omega;q) = \pi_j(\omega;\eps q, x_0)$ unless $\pi_j(\omega;\eps q, x_0)$ accesses $x_0$, in which case $\tau_j^*(\omega;q) = R_j$, and
        $\tau_j^*(\bm{\alpha}^{1:t};q) = \pi_j(\bm{\alpha}^{1:t};\eps+(1-\eps) q, x_0)$ unless this accesses $x_0$, in which case $\tau_j^*(\bm{\alpha}^{1:t};q)= R_j$. The unscreened policy is: $\tau_j(\omega;q) = \E{\pi_j(\omega;\eps q, x_0)}$ and $\tau_j(\bm{\alpha}^{1:t};q) = \E{\pi_j(\bm{\alpha}^{1:t}, \eps+(1-\eps)q,x_0}.$
\end{restatable}
\begin{proof}
    See \Cref{pf:eGpBPolicyCorrel}, p.\ \pageref{pf:eGpBPolicyCorrel}.
\end{proof}

If all but one program is an $\eps$Grounded$\pi$Bot, all players obtain the same utility as their repeated game policies do in $G_\eps$, up to a factor of $\eps$:

\begin{restatable}{thm}{correlUtils}\label{thm:correl_utils}
Let $G$ be an $n$-player correlated program game. Let $p_{-i}$ be correlated $\eps$Grounded$\pi$Bots with policies $\pi_{-i}$, and $p_i$ be any program for player $i$ such that programs $p_1,\ldots, p_n$ almost surely halt against each other. For each $j$, let $\tau_j$ and $\tau_j^*$ be the associated $G_\eps$ policies for player $j$. Then $\E{\bm{u}\mid \p} = \eps\Ex_{G_\eps}[\bm{u}\mid \bm{\tau},\bm{\tau^*}]$.
\end{restatable}
\begin{proof}
    See \Cref{pf:correlUtils}, p.\ \pageref{pf:correlUtils}.
\end{proof}

This gives us the following folk theorem:

\begin{restatable}{cor}{correlFolkThm}\label{cor:correl_folk_thm}
Suppose that $G$ is an $n$-player program game. Then for any feasible (with correlation), strictly individually rational $v_{1:n}$ there exists a program equilibrium of correlated $\eps$Grounded$\pi$Bots (with a common choice of $\eps$).
\end{restatable}
\begin{proof}
    See \Cref{pf:correlFolkThm}, p.\ \pageref{pf:correlFolkThm}.
\end{proof}

\begin{eg}\label{eg:pirates1}
Pirates 1, 2, 3 have plundered 30 gold pieces. Each pirate can Cooperate ($C$) with the other pirates, Defect ($D$) and try to run off with more than their fair share of the gold pieces, or hire a Lawyer ($L$) to enforce a fair split.
If all three pirates Cooperate, they each receive 10 gold pieces. If one or more pirates defect, and the rest cooperate, then the first up-to-two pirates in numeric order to defect receive 14 gold pieces, and all other pirates receive 0. (The defecting pirates clumsily drop the rest of the gold pieces in the sea.)
If one pirate hires a lawyer, and the other pirates cooperate, the pirate who hired the lawyer receives 9 gold pieces, and other pirates receive 10 gold pieces as usual. If at least one pirate defects, and exactly one pirate hires a lawyer, the outcome is the same as if the latter pirate cooperated (they have no money to pay the lawyer). If two pirates hire a lawyer, the lawyer enforces a fair split, and all pirates receive 9 gold pieces, regardless of what the third pirate plays.
We show the payoffs in \Cref{tab:Pirates}.

Note that $(C,C,C)$ is the unique symmetric (strictly) \emph{Pareto optimal} strategy profile (strictly Pareto optimal meaning that for any other strategy profile $(s_1,s_2,s_3)\neq (C,C,C)$, we have $u_i(s_1,s_2,s_3)<u_i(C,C,C)$ for some $i$). Can correlated $\eps$Grounded$\pi$Bots enforce a $(C,C,C)$ equilibrium, for a utility of 10 for each player? \Cref{cor:correl_folk_thm} gives that the answer is yes: The minimax utility for each player is at most 9. Hence, this payoff profile is feasible and strictly individually rational. We explicitly describe a correlated $\eps$Grounded$\pi$Bot equilibrium for this game in \Cref{app:correl_pirates}.

\begin{table}[tbh]
\small
\centering
\setlength{\extrarowheight}{0pt}
\begin{tabular}{cc|c|c|c|}
  \multicolumn{2}{c}{} & \multicolumn{3}{c}{Player 2}\\
  \multicolumn{1}{r}{Plr 3} & \multicolumn{1}{r}{Plr 1} & \multicolumn{1}{c}{$C$}  & \multicolumn{1}{c}{$D$} & \multicolumn{1}{c}{$L$}\\
\cline{3-5}
 & $C$ & $10,10,10$ & $0,14,0$ &$10,9,10$\\
$C$ & $D$ & $14,0,0$ & $14,14,0$&$14,0,0$ \\
 & $L$ & $9,10,10$  & $0,14,0$&$9,9,9$\\
\cline{3-5}
 & $C$ & $0,0,14$ & $0,14,14$ & $0,0,14$ \\
$D$ & $D$ & $14,0,14$ & $14,14,0$ & $14,0,14$ \\
 & $L$ & $0,0,14$ & $0,14,14$ & $9,9,9$ \\
\cline{3-5}
 & $C$ & $10,10,9$ & $0,14,0$ & $9,9,9$ \\
$L$ & $D$ & $14,0,0$ & $14,14,0$ & $9,9,9$ \\
 & $L$ & $9,9,9$ & $9,9,9$ & $9,9,9$ \\
\cline{3-5}
\end{tabular}
\caption{The game between the three pirates.}
\label{tab:Pirates}
\end{table}
\end{eg}
\section{Epsilon-Grounded Simulation without Shared Randomness}\label{sec:uncorrel}

We now consider uncorrelated program games, where no correlated random sequence is available to the programs, so that programs can only randomise independently.
We present our algorithm, the \emph{uncorrelated $\eps$Grounded$\pi$Bot}, in \Cref{algorithm:pi-bot-uncorrel}. It works almost identically to correlated $\eps$Grounded$\pi$Bots: it determines its time step based on the first element of the random sequence that is less than $\eps$, and then simulates the other programs at earlier time steps. %

\begin{algorithm}[tbh]
\caption{uncorrelated \eGpB}\label{algorithm:pi-bot-uncorrel}
\begin{algorithmic}
    \STATE {\bfseries Input:} Programs $p_{-i}$, sequence $(r_m)_{m=0}^\infty\in[0,1]^{\infty}$
    \STATE $T \leftarrow \min \{t:r_t <\epsilon \}$
    \FOR{$t=1$ to $T$}
        \FOR{$j=1$ to $n$}
            \STATE $a_j^t \leftarrow \apply(p_j, (p_{-j},(r_m)_{m\geq T+1-t}))$
        \ENDFOR
    \ENDFOR
    \RETURN{$\pi_i(\bm{a}^{1:T};r_{0:T})$}
\end{algorithmic}
\end{algorithm}

\paragraph{Halting Behaviour}
We obtain essentially identical halting, halting robustness, and expected runtime results as for the correlated $\pi$Bots -- the uncorrelated $\eps$Grounded$\pi$Bots halt against each other, have finite expected runtime against each other if memoization is used, and still halt if one program is replaced with a program satisfying conditions similar to those in  \Cref{thm:correl_halting_ext}. We therefore state these results formally only in \Cref{app:uncorrel_halting}.

\paragraph{No full folk theorem}

We cannot obtain the full folk theorem with the uncorrelated $\eps$Grounded$\pi$Bots. We show this as follows:

\begin{eg}\label{eg:pirates_uncorr_pi_bots}
    Recall the game between the pirates from \Cref{eg:pirates1}. Can we still obtain a $(C,C,C)$ equilibrium with uncorrelated $\eps$Grounded$\pi$Bots? Suppose that programs $p_1,p_2,p_3$ are $\eps$Grounded$\pi_i$Bots which play $C$ against one another with probability 1. We claim that this cannot be an equilibrium. To see this, consider a program $p_1'$ which looks at $r_{T+1}$ and plays $D$ if $r_{T+1}<q$, for some $q$, and otherwise plays $C$.

    Now, note that in computing the $\bm{a}^{1:t}$, the $\apply$ call tree of the $\eps$Grounded$\pi$Bots only has calls with the sequence $(r_m)$ truncated by at most $T$. Thus, $r_{T+1}$ is the same for each call in the call tree. Then, if $r_{T+1}\geq q$, which occurs with probability $1-q$, $p_i$ ($i\neq 1$) must output $C$. Moreover, the outputs of the three programs are independent. Hence, player 1 receives utility at least $\min_{s_2,s_3} u_1((1-q)C+qD, (1-q)C + qs_2, (1-q)C + q s_3)$.
    Now, when minimising the above quantity with respect to $s_2$ and $s_3$, we may assume WLOG that $s_2(C)=s_3(C) = 0$, since players 2 and 3 playing $C$ is always at least as good as other actions for player 1. Moreover, a mixture of $D$ and $L$ is always no better than both players using the better one of the two (see \Cref{app:pirates_uncorr_pi_bots} for a derivation of this). We then need only consider $s_2=s_3=L$ and $s_2=s_3=D$. %
    Taking $q=\frac{3}{4}$ and considering both possibilities, we obtain that player 1 receives utility equal to the minimum of $10\frac{3}{4}> 10$ and $10\frac{21}{32} > 10$. Thus, $(p_1,p_2,p_3)$ cannot be an equilibrium.
\end{eg}

The above example illustrates a general problem: if an opponent randomises in such a way that their actions at different time steps in the $\eps$Grounded$\pi$Bots' simulations are perfectly correlated, deviation with probability $q$ can be detected with probability at most $q$. Since actions are then uncorrelated between programs, a player may be able to benefit from deviating.

\paragraph{Characterising the equilibria}
We now characterise the equilibria that may be obtained by uncorrelated $\eps$Grounded$\pi$Bots.

\begin{restatable}{prop}{uncorrelEquilibria}\label{prop:uncorrelEquilibria}
Suppose that payoffs $v_{1:n}$ are feasible (without correlation), with $v_{1:n} = u_{1:n}(s_{1:n})$. Suppose further that for each $i$ we have $v_i \geq \min_{(S_{j}':\mathcal{A}_i\rightarrow\Delta(\mathcal{A}_j))_{j\neq i}} \max_{q\in[0,1], s_i' \in \Delta(A_i)}\allowbreak u_i\left((1-q)s_i+qs_i', \left((1-q)s_j + qS_j'\bullet s_i'\right)_{j\neq i}\right) -\lambda q$, for some $\lambda > 0$, writing $S_j'\bullet s_i'$ for $\sum_a s_i'(a)S_j'(a)$. Then there exist uncorrelated $\eps$Grounded$\pi$Bot equilibria with payoffs $v_{1:n}$.
    
\end{restatable}
\begin{proof}
    See \Cref{pf:uncorrelEquilibria}, p.\ \pageref{pf:uncorrelEquilibria}.
\end{proof}

The main ideas for the above are as follows: We can approximate any strategy $s_i$ arbitrarily closely via a repeating sequence of pure actions over the different time steps. As $\epsilon$ goes to zero, the chance of detecting a deviation from this sequence is almost as large as the probability of said deviation. Program $j$ then plays a punishment strategy ($S_j'$) according to the observed deviation action.

Using similar ideas to \Cref{eg:pirates_uncorr_pi_bots}, we obtain a converse:
\begin{restatable}{prop}{uncorrelEquilibriaNecessary}\label{prop:uncorrelEquilibriaNecessary}
    Suppose that there exists an uncorrelated $\eps$Grounded$\pi$Bot equilibrium with payoffs $v_{1:n}$. Then $v_{1:n}$ must be feasible (without correlation), and for each $i$, we must have $
        v_i \geq \min_{(S_{j}':\mathcal{A}_i\rightarrow\Delta(\mathcal{A}_j))_{j\neq i}} \max_{q\in[0,1],s_i' \in \Delta(\mathcal{A}_i)}\allowbreak
        u_i\left((1-q)s_i+qs_i', \left((1-q)s_{-i} + qS_j'\bullet s_i'\right)_{j\neq i}\right)$.
\end{restatable}
\begin{proof}
    See \Cref{pf:uncorrelEquilibriaNecessary}, p.\ \pageref{pf:uncorrelEquilibriaNecessary}.
\end{proof}
We also obtain the following corollary (by considering small probability deviations in the above result):
\begin{cor}\label{cor:uncorrel_necessary}
    Suppose that $v_{1:n} = u_{1:n}(s_{1:n})$ in some game $G$. Then if for some $i$, there exists $s_i'\in \Delta(\mathcal{A}_i)$ such that $
        \min_{(s_j'\in \Delta(\mathcal{A}_i))_{j\neq i}} \sum_j u_i(s_j',s_{-j}) - v_i > 0$,
    then $v_{1:n}$ is not attainable in any uncorrelated $\eps$Grounded$\pi$Bot equilibrium.
\end{cor}
    Since the behaviour of the uncorrelated $\eps$Grounded$\pi$Bots is unaffected by changes to a program that do not affect the program's output against the other programs on a given input, the equilibria are also robust to such changes. As a special case, if an uncorrelated $\eps$Grounded$\pi$Bot equilibrium entails one program playing a pure strategy at equilibrium, then replacing that program with any other program that almost surely plays that same pure strategy against the $\eps$Grounded$\pi$Bots will leave the expected utilities of all players unchanged.

Even in two-player games, our $\eps$Grounded$\pi$Bots retain an advantage over the \citet{oesterheld2019robust} $\eps$Grounded$\pi$Bots. In particular, they allow for better detection and punishment of deviation from mixed strategies, which allows for more equilibria. This is illustrated in the following example:

\begin{table}[ht]
\centering
\setlength{\extrarowheight}{3pt}
\begin{tabular}{cc|c|c|}
  & \multicolumn{1}{c}{} & \multicolumn{2}{c}{Player $2$ }\\
  & \multicolumn{1}{c}{} & \multicolumn{1}{c}{$G$}  & \multicolumn{1}{c}{$C$}\\
  \cline{3-4}
  \multirow{2}*{Player 1}  & $K$ & $3,0$ & $3,0$\\
  \cline{3-4}
  & $S$ & $2,4$ & $4,2$\\
  \cline{3-4}
\end{tabular}
\caption{The trust game from \Cref{eg:trust}}
\label{tab:TrustMixed}
\end{table}

\begin{eg}\label{eg:trust}
Consider the following trust game (cf.~\citet{Berg95:Trust};
with payoffs shown in \Cref{tab:TrustMixed}): Player 1 receives \$3, which he may Keep ($K$) or Send to the other player ($S$). Money sent to the other player is doubled. Player 2 may then either be Greedy ($G$) and take $\$4$ and send only $\$2$ back, or be Charitable ($C$) and take only $\$2$, giving $\$4$ to player 1. We show in \Cref{app:mixed} that there exists an uncorrelated $\eps$Grounded$\pi$Bot equilibrium in which player 1 plays $S$ and player 2 randomises at an approximately $1:1$ ratio between $T$ and $R$, whereas the \citet{oesterheld2019robust} $\eps$Grounded$\pi$Bot cannot attain any equilibria in which player 1 plays $S$ and player 2 mixes.
\end{eg}

In general, there is no simple relationship between utilities in games of $\eps$Grounded$\pi$Bots and repeated games. We discuss this further in \Cref{app:rep_uncorrel}. In the special case of additively separable utilities, we do obtain a folk theorem:

\begin{restatable}{thm}{additiveFolk}\label{thm:additiveFolk}
    Suppose $G$ is an $n$-player program game with additively separable utilities. Then for any set of payoffs $v_{1:n}$, there exists a program equilibrium of uncorrelated $\eps$Grounded$\pi$Bots (with a common $\eps$) whenever $v_{1:n}$ is feasible (without correlation) and strictly individually rational.
\end{restatable}
\begin{proof}
    See \Cref{pf:additiveFolk}, p.\ \pageref{pf:additiveFolk}.
\end{proof}

\begin{eg}
    Consider \Cref{eg:intro_game} from the introduction.
    Can the uncorrelated $\eps$Grounded$\pi$Bots attain a $(C,C,C)$ equilibrium? Since the game has additively separable utilities, \Cref{thm:additiveFolk} gives that they can. In particular, we may use the following policies: $\pi_1(\bm{a}^{1:t}) = C$ if no player has defected, and $P_i$ if player $i$ is the first to defect. For $i>1$, use $\pi_i(\bm{a}^{1:t}) = C$ if no player has defected and $D$ otherwise. Player 1 cannot benefit by deviating. Meanwhile, if player 2, say, defects with probability $q$, the defection will be observed by each other player with probability at least $q(1-\eps)$ (independently between players). Thus, the gain from defecting will be $5q - (3+3)q(1-\eps)$. Thus, this is an equilibrium provided $\eps<1/6$.
\end{eg}

\section{General Simulation-based Programs}\label{sec:simulationist}
We now turn to more general simulation-based programs, in the uncorrelated setting from the last section. Can we do better than the uncorrelated
$\eps$Grounded$\pi$Bots? We start by defining the class of programs we are interested in.

Say a program $p_i$ is \emph{simulationist} if it only references the input programs, $p_{-i}$, via calls of the form $\apply(p_j,(p_{-j}',(r_m')))$ where each program in $p_{-j}'$ is either one of the input programs, or is itself simulationist, $(r_m')$ is any sequence, and we allow $i=j$.

This definition is inductive -- e.g., uncorrelated $\eps$Grounded$\pi$Bots are simulationist, as are, say, DefectBots. Meanwhile, a program $p$ that ran the input programs on a mix of such programs and input programs would be simulationist, as would a program $q$ that ran the input programs on $p$. Note that we exclude programs that run the input programs on programs that are \emph{not} simulationist. We explain this further in \Cref{app:simulationist}.

We give here a simple example of a game in which it is possible to sustain a profile of equilibrium payoffs with \emph{some} simulationist programs, but not with uncorrelated $\eps$Grounded$\pi$Bots.

\begin{table}[ht]
    \centering
    \setlength{\extrarowheight}{3pt}
    \begin{tabular}{cc|c|c|}
      & \multicolumn{1}{c}{} & \multicolumn{2}{c}{Player $2$ }\\
      & \multicolumn{1}{c}{} & \multicolumn{1}{c}{$G$}  & \multicolumn{1}{c}{$F$}\\
      \cline{3-4}
      \multirow{2}*{Player 1}  & $K$ & $3,0$ & $3,0$\\
      \cline{3-4}
      & $S$ & $0,6$ & $4,2$\\
      \cline{3-4}
    \end{tabular}
    \caption{The trust game from \Cref{eg:trust_pure}.}
    \label{tab:TrustSimple}
\end{table}

\begin{eg}\label{eg:trust_pure}
    Consider the following trust game: Player 1 receives $\$3$, which she may Keep ($K$) or Send to the other player ($S$). Money sent to Player 2 is doubled. Player 2 may then be Greedy ($G$) and keep all $\$6$, or Fair ($F$) and return $\$4$ to player 1. The payoff matrix is shown in \Cref{tab:TrustSimple}. Can an $(S,F)$ equilibrium be sustained?
    
    For the uncorrelated $\eps$Grounded$\pi$Bots, \Cref{cor:uncorrel_necessary} shows that the answer is no:  Note that $u_2(S,F) = 2$. Let $s_2' = G$. Then, we have $\min_{s_1'\in \Delta(\mathcal{A}_1)}\allowbreak (u_2(s_1',F)-v_2) + (u_2(S,G)-v_2) = u_2(K,F)-2+u_2(S,G)-2 = 2 > 0$.

    However, an $(S,F)$ equilibrium can be sustained as follows: Player 2 submits program $p_2$, which just outputs $F$. Player 1 submits program $p_1$, which simulates the other program twice on independent input sequences (e.g., by splitting its input sequence into odd and even entries), and then outputs $S$ only if both simulations output $F$. Since $p_2$ never calls any other program, $p_1$ and $p_2$ halt against each other, with outputs $(S,F)$. It is easy to show that this is then an equilibrium.
\end{eg}

The above example is relatively extreme, but illustrates two ideas for outperforming uncorrelated $\eps$Grounded$\pi$Bots. First, taking multiple samples. Typically, we cannot take multiple independent samples of a program's action without halting issues. In the previous example, we could, because of program 2 always being effectively on time step 1. In general, we can take multiple samples of a program's action, provided that we ensure all samples are on earlier time steps than our own.
Second, using different distributions over the time step for different players. In the above example, there is no need for player 2 to simulate player 1, so we may have player 2 be always on time step 1. More generally, we might do better in some games by having players have different (and potentially non-geometric) distributions over their time step. We discuss these ideas further in \Cref{app:gen_simul_idea}.

\paragraph{Negative results}
We will now explore the limits of equilibria enforceable with simulationist programs. We first need to place a mild additional constraint on the programs we consider:

Say that programs $p_1,\ldots,p_n$ halt against each other \emph{output independently} if they almost surely halt against each other, and this remains true if one program, $p_i$, is replaced by a program $p_i'$ that almost surely halts on any input that $p_i$ almost surely halts on. This condition rules out programs that, e.g., loop forever once they determine that another program has played the wrong action. In particular, if programs halt against each other output independently, they must still halt if one program is replaced with a program that simply plays according to its minimax strategy, and so individually rational payoffs is a necessary condition for equilibria of such programs. $\eps$Grounded$\pi$Bots halt output independently against each other (\Cref{app:gen_simul_further}, \Cref{prop:output_indep_halting}).

We obtain the following negative result. An equivalent result expressed explicitly in terms of payoffs may be found in (\Cref{app:gen_simul_further}, \Cref{thm:simulationist_negative_result}).

\begin{restatable}{thm}{simulationistNegativeResultGame}\label{prop:simulationistNegativeResultGame}
    Consider games of the following form, given $\mathbb{N}$-valued random variables $T_1,\ldots T_n$: Each player $i$ observes $T_i$, sampled independently from other players, and for each $t<T_i$ sees each other player $j$'s policy conditional on that player observing $T_j = t$, and then chooses a strategy.
    
    If an equilibrium with payoffs $v_{1:n}$ is attained by simulationist programs that halt output independently, there must exist a policy profile $\tau_{1:n}$ in a game of the above form attaining payoffs $v_{1:n}$, with the following property: for all $i$, all possible deviations of the form `play according to $\tau_i$ if $T_i<t_0$, else play according to $d_i$' for some fixed strategy $d_i$ attain at most as much utility against $\tau_{-i}$ as does $\tau_i$.
\end{restatable}
\begin{proof}
    See \Cref{pf:simulationistNegativeResultGame}, p.\ \pageref{pf:simulationistNegativeResultGame}.
\end{proof}

The idea behind this result is as follows: for any profile of simulationist programs $p_{1:n}$, we may effectively define a `time step' $T_i$ for each program $p_i$, depending on the random sequence, as the maximum depth of the simulation tree with input programs $p_{-i}$. We then consider programs $p_i'$ that look at the `time step' $T_i$, and deviate if it is sufficiently high (at least $t_0$, say). We can show that the other programs only observe this deviation when their own time step is in turn high enough (at least $t_0+1$). This result then also provides additional justification for the idea of a simulated repeated game, since it shows that all simulationist programs may be in some sense viewed as being on a particular (random) time step and only being able to observe deviations occurring on earlier time steps.

From this result, we obtain that in the pirate example (\Cref{eg:pirates1}) the action profile $(C,C,C)$, cannot be attained in a simulationist uncorrelated program equilibrium (see \Cref{eg:pirates_gen_simul}, below). Thus, the folk theorem does not hold for simulationist programs in the uncorrelated setting. We further discuss which equilibria may be attained in \Cref{app:gen_simul_further}.

\begin{eg}\label{eg:pirates_gen_simul}
    Recall the game between the pirates from \Cref{eg:pirates1}.
Can the pirates obtain pure cooperation in a simulationist equilibrium? We will apply \Cref{prop:simulationistNegativeResultGame} to show that the answer is no. Consider any choice of random variables $T_1$, $T_2$, $T_3$ for the time steps of the three players. Let $t_0$ be maximal such that $\prob{T_1\geq t_0} \geq \frac{3}{4}$, and suppose pirate 1 defects whenever she observes $T_1 \geq t_0$. There are two potentially optimal punishment strategies: punishing with $D$, and punishing with $L$ (cf.\ \Cref{app:pirates_uncorr_pi_bots}). A $D$ punishment strategy cannot work here, since even if both other players were to always play $D$, player 1 would still receive a utility of $\frac{3}{4}\times 14 > 10$. So we may assume that the other players play $L$ when they observe player 1 defect. Now, player 2 observes player $1$ defect with probability $\prob{T_2\geq t_0+1}$. Even if player 3 were to always observe the defection, for the punishment to be effective, we then need $4\prob{T_1\geq t_0}(1-\prob{T_2\geq t_0+1}) - \prob{T_2\geq t_0+1} \leq 0$, i.e. $\prob{T_2\geq t_0+1} \geq \frac{4\prob{T_1\geq t_0}}{4\prob{T_1\geq t_0}+1} \geq \frac{3}{4}$. Then instead consider the case where pirate 2 defects whenever he observes $T_2 \geq t_0+1$. Pirate 1 sees this defection with probability $\prob{T_1\geq t_0+2}$, and we may obtain by similar argument that we then need $\prob{T_1\geq t_0+2} \geq \frac{4\prob{T_2\geq t_0+1}}{4\prob{T_2\geq t_0+1}+1} \geq \frac{3}{4}$. But this contradicts the definition of $t_0$.

\end{eg}

\section{Related Work}\label{sec:related_work}

Our work is most closely related to papers on program equilibrium, especially \citet{oesterheld2019robust}, to which we have related our ideas throughout the paper. \citet{digiovanni2023commitment} use a variant of $\epsilon$GroundedFairBot to implement mutual commitments to disclose information. Their extension addresses similar limitations of \citeauthor{oesterheld2019robust}'s program as ours (failures to halt when simulating multiple opponents). Like the present paper, their main idea is based on correlated halting, somewhat similarly to our \Cref{sec:correl}, but via access to a shared simulation device. They prove a variant of the full folk theorem for their setting. However, their approach differs from our correlated $\eps$Grounded$\pi$Bots in important ways. For example, they assume that programs directly observe each others' randomisation probabilities, and their programs only allow for a restricted range of policies. We give a detailed comparison in \Cref{app:comparison-to-digiovanni-et-al}. 

A recent line of papers has studied games in which players are specifically given the ability to simulate each other. Most relevantly, \citet{kovarik2024recursive} give players a joint simulation device, which allows simulating an interaction between the players. With probability $1-\epsilon$, the simulated players in turn have access to a joint simulation device, and so on. Very roughly, our correlated $\epsilon$Grounded$\pi$Bots could be viewed as decentralised implementations of such joint simulation. Like our $\epsilon$Grounded$\pi$Bots, joint simulation devices enable a full folk theorem. \citet{kovarik2023game} study games in which only one player is able to simulate the other at some cost, which is technically much less related.

\section{Conclusion}
The concept of program equilibrium promises to enable one-shot cooperation among agents whose code is legible to each other.  Among several approaches to designing bots that play program equilibria, 
\citeauthor{oesterheld2019robust}'s \citeyearpar{oesterheld2019robust} $\eps$Grounded$\pi$Bot stands out for its simplicity and robustness, requiring only simulation of the other program.  Nevertheless, the original design of has several limitations, facing challenges in generalising to settings with more than two players, and falling short of the full folk theorem even with two players.

In this paper, we considered how $\eps$Grounded$\pi$Bot might be generalised to address these issues.  The key insight was to have a program determine its own `time step' first and then manipulate simulated programs' sources of randomness so that they must be on earlier time steps.  With this idea, we first showed that if the programs have access to shared randomness, we obtain a full folk theorem.  We also characterised the equilibria that can be obtained without shared randomness, but these fall short of the full folk theorem.  We showed one could do a bit better with more general `simulationist' programs, but these still fall short of the full folk theorem as well.

\section*{Acknowledgments}

We are grateful to the anonymous AAAI reviewers for helpful comments that improved the exposition of this paper. We also thank the Cooperative AI Foundation, Polaris Ventures (formerly the Center for Emerging Risk Research) and Jaan Tallinn’s donor-advised fund at Founders Pledge for financial support. Caspar Oesterheld was supported by an FLI PhD fellowship. Emery Cooper began work on this project while working for the Center on Long-Term Risk.

\bibliography{refs}

\appendix
\onecolumn
\section{Why the \citet{oesterheld2019robust} $\eps$Grounded$\pi$Bot cannot be na\"ively extended to 3 or more players}\label{app:2019_halting}
The obvious way to extend the \citet{oesterheld2019robust} $\eps$Grounded$\pi$Bot to games of 3 or more players would be for it to simply simulate both players with probability $1-\eps$. However, it is easy to see that this will not work, at least if its opponents are also $\eps$Grounded$\pi$Bots. First, this would result in an expected $2(1-\eps)$ depth 1 simulations (i.e., simulations made by a non-simulated program), which would lead to $4(1-\eps)^2$ expected depth 2 simulations (i.e., simulations made directly by a depth 1 simulation), and so on, for an infinite expected runtime when $\eps\leq 0.5$. More generally, if the expected number of depth 1 simulations were $N$ (accounting for probability of zero such simulations being run), then the expected runtime would be infinite whenever $N \geq 1$.

Worse still, if two simulations are run with probability $1-\eps$, the halting probability is less than one for small $\eps$. Let $f(t)$ be the probability the maximum simulation depth is at most $t$ (where is no simulations are made, the depth is 0). Then $f(0) = \eps$ and for $t\geq 1$, we have $f(t) = \eps + (1-\eps)f(t-1)^2$. Hence, if $f(t-1) \leq \frac{\eps}{1-\eps}$, we have $f(t)\leq \eps + \frac{\eps^2}{1-\eps}= \frac{\eps}{1-\eps}$. So since $f(0) = \eps < \frac{\eps}{1-\eps}$, we obtain that $f(t)\leq \frac{\eps}{1-\eps}$ for all $t$. Thus, the probability of halting is $\lim_{t\rightarrow\infty}f(t)\leq \frac{\eps}{1-\eps}< 1$ for $\eps<\frac{1}{2}$.

Similarly, if instead the number of simulations is a random variable $N$, with $\prob{N=k} = a_k$, we would have $f(t) = \sum_{k=0}^\infty a_k f(t-1)^k$ for $t\geq 1$, and $f(0) = a_0$. Now, consider $g(x) = \sum_{k=0}^\infty a_k x^k$. We then have that $g$ is increasing, with $g(1) = 1$, and $g'(x) = \sum_{k=1}^\infty k a_k x^{k-1}$. Thus, $g'(1) = \E{N}$. Then, if $\E{N}>1$, $g'(x)$ must be greater than 1 on some interval $[x_0,1]$ with $x_0< 1$ (by continuity of $g'$), and so by the mean value theorem, $g(x)<x$ for $x_0\leq x < 1$. WLOG take $x_0> a_0$. Then, combined with the fact $g$ is increasing, this implies that $g^n(a_0)$ can never exceed $x_0< 1$ when $\E{N}>1$. Thus, $f(t)\leq x_0<1$ for all $t$, and so the probability of halting is less than 1 whenever the expected number of simulations directly made by each $\eps$Grounded$\pi$Bot is greater than 1.

\section{How the \citet{oesterheld2019robust} $\eps$Grounded$\pi$Bot can solve the 3 player prisoner's dilemma, and why this doesn't work for \Cref{eg:intro_game}}\label{app:2019_3player_schemes}
Despite the fact that the \citet{oesterheld2019robust} $\eps$Grounded$\pi$Bot cannot make more than 1 expected simulation (cf.\ \Cref{app:2019_halting}), it can still be effective in some games of more than 2 players. Sometimes, one simulation is enough. For instance, consider a three-player prisoner's dilemma: Each player chooses between Cooperate ($C$) and Defect ($D$). For each player that cooperates, every player receives an additional 1 utility. If a player defects, she instead keeps 2 utility for herself. Thus, the utility of player $i$ is equal to $(\#\text{ of players that cooperate}) + 2I\{\text{player }i\text{ defected}\}$.

Then the \citet{oesterheld2019robust} $\eps$Grounded$\pi$Bot may simply simulate a random other player, and play the same action as that player. If player 1, say, then deviates to overall strategy $(1-q)C+qD$, other players will defect with probability $p$, where $p = \frac{(1-\eps)}{2}q + \frac{(1-\eps)}{2}p$, since with probability $\frac{(1-\eps)}{2}$ player 2 (say) will simulate player 1 directly, and with probability $\frac{(1-\eps)}{2}$ he will instead simulate player 3, who also defects with probability $p$, by symmetry. Thus, solving for $p$, players 2 and 3 defect with probability $\frac{1-\eps}{1+\eps}q$. The profit to player 1 of deviating is then equal to $q - 2\frac{1-\eps}{1+\eps}q$, which is non-positive when $\eps \leq \frac{1}{3}$. Thus, this is an equilibrium for $\eps\leq \frac{1}{3}$. We can also use schemes in which, e.g., player 1 simulates player 2, player 2 simulates player 3, and player 3 simulates player 1. These schemes works because there is no need for players to know who first defected, since they can defect to effectively punish everyone.

By contrast, in \Cref{eg:intro_game}, player 1 needs to know which player defected first, since the punishment actions just transfer utilities between players. Moreover, player 2 and 3 also need to punish each other for defecting, so schemes in which only player 1 simulates the other players (cf.\ \Cref{eg:trust_pure}) will not work. Formally, suppose that player 1 were to simulate player 2 with probability $\rho_{12}$ and player 3 with probability $\rho_{13}$, where $\rho_{12}+\rho_{13} < 1$. Now, players 2 and 3 have no need to simulate player 1 (since player 1 has no incentive to deviate from $(C,C,C)$), and so WLOG assume that player 2 simulates player 3 with probability $\rho_2$, and player 3 simulates player 2 with probability $\rho_3$, where $\rho_2, \rho_3 < 1$. Now, assume (WLOG, since can otherwise decrease $\rho_2$ or $\rho_3$) that players 2 and 3 defect when the simulation they run defects.

Now, suppose that player 2 deviates to defecting with probability $q$. Then player 3 will defect with overall probability $\rho_3 q$. Thus, player 1 would observe player 2 defecting with probability $\rho_{12}q$, and player 3 defecting with probability $\rho_{13}\rho_3 q$. Now, there is no point in player 1 randomising which player to punish on a given observation, since $P_2$ and $P_3$ cancel each other out. Moreover, we may assume WLOG that player 1 plays either $P_2$ or $P_3$ when one of player 1's simulations defect, otherwise we may decrease the simulation probabilities. We can then consider 4 different policies: always punishing player 2 for defection, always punishing player 3, always punishing the player observed defecting, or always punishing the other player. We can see immediately that the first two policies won't work (in the first case by considering player 3 defecting), so we can restrict attention to (i) always punishing the player observed defecting, and (ii) always punishing the other player. In case (i), player 1 will then play $P_2$ with probability $\rho_{12}q$ and $P_3$ with probability $\rho_{13}\rho_3 q$. The gain from deviating for player 2 is then $5q - 3(\rho_{12}-\rho_{13}\rho_3)q - 3\rho_3 q$, and so we would need $3\rho_3 + 3(\rho_{12}-\rho_{13}\rho_3) \geq 5$, i.e., $3\rho_{12}+ 3\rho_3(1-\rho_{13})\geq 5$. For this, we need $\rho_{12} -\rho_{13}\geq \frac{2}{3}$. But then considering player 3 defecting, we would need $\rho_{13}>\rho_{12}$, a contradiction. Thus, consider instead policy (ii). We then get that player 2 profits from deviating by $5q + 3(\rho_{12}-\rho_{13}\rho_3)q - 3\rho_3q$. We would then need $3\rho_3(1+\rho_{13})-3\rho_{12}\geq 5$, and so $\rho_{13} -\rho_{12}\geq \frac{2}{3}$. Considering deviations for player 3, this fails similarly.

\section{A detailed comparison to \citeauthor{digiovanni2023commitment}'s construction}
\label{app:comparison-to-digiovanni-et-al}

We here give a detailed comparison to \citeauthor{digiovanni2023commitment}'s (\citeyear{digiovanni2023commitment}) extension of \citeauthor{oesterheld2019robust}'s (\citeyear{oesterheld2019robust}) $\eps$GroundedFairBots. The comparison is complicated by the fact that \citeauthor{digiovanni2023commitment} consider a different setting, namely one where players can choose whether to (credibly) disclose their private information. Translated to our setting, though, here's a rough description of their extension. Imagine the programs have access to some shared randomisation device $\mathcal C ^ P$. Let's say that the goal is to implement some strategy profile $\boldsymbol \sigma$. Then their proposed program works roughly as follows: With probability $\epsilon$ -- where the randomness is sampled from $\mathcal C ^ P$ -- terminate and play according to $\boldsymbol{\sigma}$. Otherwise, simulate all opponents. If all opponents play according to $\boldsymbol{\sigma}$, then play according to $\sigma_i$ as well. Otherwise play the minimax strategy against the appropriate player. (There are some subtleties to determining who the deviating player is (since multiple opponent simulations may be found to deviate from $\boldsymbol{\sigma}$), but those subtleties don't matter here.) Halting is ensured by how $\mathcal C ^ P$ works. In particular, when $\mathcal C ^ P$ is called by one of the programs, $\mathcal C ^ P$ can see the depth of the call stack of that program. $\mathcal C^P$ correlates its answers across different programs for a given depth. This way, the programs all terminate at the same depth.

Their approach is most similar to our correlated $\epsilon$Grounded$\pi_i$Bots. There are two main similarities between how these two approaches extend \citeauthor{oesterheld2019robust}'s work. 1) Like our paper, they use correlated termination to enable simulation of multiple opponents. 2) They prove a folk theorem for their bots.

However, there are important differences between their approach and ours. The most important difference is that their setting takes care of some of the issues that we have to deal with in our programs. For one, $\mathcal C ^ P$ ``automatically'' takes care of correlated halting, which our programs implement by how they feed each other the sequences of random numbers $r$ and $x$. More significantly, their programs require that programs return the mixed strategies that they play and that one can directly access those mixed strategies. In contrast, we only ever sample a \textit{single} pure strategy/action from applying the opponent program.

Because of the latter difference, it's enough for \citet{digiovanni2023commitment} to only simulate each opponent once (`on the previous time step') and react in only one of two different ways (play $\boldsymbol\sigma$ or play minimax against the deviating player). Note specifically that in the two-player case (if we remove correlated halting, which is unnecessary in the two-player case), \citeauthor{digiovanni2023commitment}'s program is very similar to \citeauthor{oesterheld2019robust}'s $\epsilon$GroundedFairBot, where $\sigma_i$ is used in place of Cooperate and instead of Defect we use the minimax strategy against the opponent.
In contrast, our programs much more closely relate to repeated games with full memory. We thus obtain results like \Cref{lemma:eGpB_policy_correl,thm:correl_utils} that specifically relate $\eps$GroundedFairBots to repeated play.

\section{Obtaining correlated randomness via randomising over programs}\label{sec:princip_rand}
One way of obtaining correlated random sequences might be for the principals (that is, the players who submit the programs) to randomise over a set of programs, each encoding a different random sequence. Then, the programs could add their random sequences together (modulo 1), in order to obtain a shared random sequence. If all other players randomise uniformly at random (and submit $\eps$Grounded$\pi$Bots that obtain the shared random sequence in this way), no player can benefit from submitting any other distribution of sequence, since the overall sequence will still be uniform.

In practice, the principals cannot actually encode an infinite sequence in their programs. However, a similar effect could be achieved via each principal randomising over a random seed which they include in their program. The programs could then combine each of their random seeds in some way and then use this to generate a shared pseudorandom sequence. For instance, the random seeds might be strings of $k$ bits, and they could then be combined by taking the XOR of all the random seeds. Then, if each principal randomises uniformly, a principal cannot change the distribution of the final random seed by deviating from this distribution.

One issue that arises from generating a random sequence from a seed of finite length is that there must then be a maximum time-step, and a program could defect if it knew it were on the maximum time-step, because it could not then be in a simulation. The hope would be that for large enough $k$, it would be very difficult to exploit this in practice, because computing all possible random sequences would be very costly for the principals. Moreover, the potential gain from exploiting this would go to zero as $k$, and hence the maximum time-step, go to infinity. Exploring these ideas further is left for future work.

\section{Statement of halting and runtime results for uncorrelated $\eps$Grounded$\pi$Bots}\label{app:uncorrel_halting}

 \begin{restatable}{thm}{haltingUncorrel}\label{thm:halting_uncorrel}
Suppose for $i = 1, \ldots n$, $\eps_i > 0$, $\pi_i$ is a computable function and $p_i$ is an uncorrelated $\eps_i$Grounded$\pi_i$Bot (\Cref{algorithm:pi-bot-uncorrel}). Suppose further that $r_0, r_1, \ldots \simiid U[0,1]$. Then for all $i$, $p_i$ halts with probability 1 on input $(p_{-i},(r_m))$.
 \end{restatable}

 \begin{restatable}{thm}{uncorrelHaltingExt}\label{thm:uncorrel_halting_ext}
Suppose that program $p_i$ has the following properties:
\begin{enumerate*}[label = (\roman*)]
    \item It never makes a standard $\apply$ call to itself.
    \item Its standard $\apply$ calls all have input sequence of the form $(r_m)_{m\geq k}$ with $0\leq k \leq c$ for some constant $c$.
    \item $p_i$ almost surely makes only finitely many standard $\apply$ calls directly (i.e., not counting those made inside other $\apply$ calls).
    \item $p_i$ almost surely halts whenever all its standard $\apply$ calls halt.
    \item For all events $A$ with $\prob{((r_m),(x^i_m))\in A} = 0$, $p_i$ almost surely makes no standard $\apply$ calls with input sequences in $A$.
\end{enumerate*}
Then if for each $j\neq i$, $p_j$ is an (uncorrelated) $\eps_j$Grounded$\pi_j$Bot with $\eps_j>0$, programs $p_{1:n}$ almost surely halt against each other.
 \end{restatable}

\begin{restatable}{thm}{expectedRuntime}\label{thm:expected_runtime}
Suppose that program $p_i$ is an uncorrelated $\eps_i$Grounded$\pi_i$Bot (\Cref{algorithm:pi-bot-uncorrel}) for $i = 1,\ldots n$. Then if at least two $p_i$ have $\eps_i<1/2$, every $p_i$ with $\eps_i < 1$ has infinite expected runtime on input $(p_{-i},(r_m))$, where $r_1,r_2,\ldots \simiid U[0,1]$. However, with memoization, the expected runtime is finite (and polynomial in $\frac{1}{\min_i \eps_i})$, provided that for all $i$, $\pi_i(\bm{a}^{1:t})$ runs in polynomial time in $t$.
\end{restatable}

\section{Relating uncorrelated $\eps$Grounded$\pi$Bots to repeated games}\label{app:rep_uncorrel}
In general, there is no simple relationship between utilities in games of $\eps$Grounded$\pi$Bots and utilities in repeated games. Here, we will restrict attention to a special case where there is a simple relationship.

We may define associated policies for repeated games with correlation (but no screening) exactly analogously as for the correlated setting, treating the sequences $(r_m^i)$ as if they were the shared sequence:
\begin{defn}[Associated repeated game policy]
    Given programs $p_1,\ldots p_n$ for game $G$, we define the associated (correlated) repeated game policy in $G_\eps$ for player $i$ as $\tau_i(\omega;q) \defeq \E{\apply(p_i,(p_{-i},(r_m)))\mid r_0=\eps q}$ and
                $\tau_i(\bm{\alpha}^{1:t};q) \defeq \Ex[\apply(p_i,(p_{-i},(r_m)))\mid r_0 = \eps + (1-\eps)q,\ \bm{a}^{1:t}=\bm{\alpha}^{1:t},T=t]$ wherever $\bm{a}^{1:t}=\bm{\alpha}^{1:t}$ occurs with positive probability.
\end{defn}
Since the sequence $(x_m)$ no longer exists, we do not need to consider screening. Similarly to before, we have that if $p_i$ is an $\eps$Grounded$\pi_i$Bot its repeated game policy is given by $\tau_i(\omega;q) = \pi_i(\omega;\eps q)$ and $\tau_i(\bm{\alpha}^{1:t},1q) = \pi_i(\bm{\alpha}^{1:t}, \eps+(1-\eps)q)$.

Now, note that the distribution of each program's output does not depend on whether the other programs receive the same random sequence as it or not. Thus, we may simply apply our result for correlated program games in the special case where no program ever accesses the $(x_m^i)$:

\begin{restatable}{cor}{uncorrActDist}\label{cor:uncorrActDist}
    Let $G$ be an $n$-player program game with uncorrelated random sequences. Let $G_\eps$ be the corresponding correlated repeated game. Suppose that for $j\neq i$, $p_j$ is an $\eps$Grounded$\pi_j$Bot, and that $p_i$ is any program such that $p_1,\ldots p_n$ almost surely halt against each other.

    For each $j$, let $\tau_j$ be the associated repeated game policy for program $j$, defined arbitrarily anywhere it is undefined. Then program $j$ has overall strategy $\E{\apply(p_j,(p_{-j},(r_m)))}$ equal in distribution to player $j$'s last action in the repeated game between policies $\bm{\tau}$.

     As a special case, if the utilities are additively separable we have $\Ex_G[\bm{u}\mid p_{1:n}]= \eps \Ex_{G_\eps}[\bm{u}\mid \tau_{1:n}]$.
\end{restatable}
\begin{proof}
    Note that in the proof of \Cref{thm:correl_utils}, we showed that the programs' actions were distributed as the actions of the repeated game policies on the last time-step of the repeated game with correlation and screening $G_\eps$. Applying this result individually for each program, in the special case where programs cannot use the sequences $(x^i_m)$, we obtain the result on overall strategies of each program.

    For the last part, we need only note that if the utilities are additively separable, the joint distribution of actions does not matter (since we then have $\E{u_{i}(A_1,\ldots,A_n)} = \E{u_{i1}(A_1)+\cdots+u_{in}(A_n)} = \sum_{j}\E{u_{ij}(A_j)}$ for any random variables $A_1,\ldots,A_n$). Hence, the same utility is obtained whether the actions are independent or correlated.
\end{proof}
This gives us the following folk theorem for games with additively separable utilities:
\additiveFolk*
\begin{proof}\label{pf:additiveFolk}
    This follows from \Cref{cor:uncorrActDist}, and our folk theorem for correlated repeated games with screening (\Cref{thm:repFolkThm}), applied to the special case where screening never occurs, and our programs cannot randomise independently. Note that in our proof of the latter, we used independent randomisation for punishment, but could equally have used the correlated random variables for this randomisation in the case of additively separable utilities, since correlation does not matter for such games.
\end{proof}

\section{Why define simulationist that way?}\label{app:simulationist}
Recall that a program $p_i$ is simulationist if it only references the input programs via calls of the form $\apply(p_j,(p_{-j}',(r_m')))$ where each $p_j'$ is either one of the input programs, or is itself simulationist. Why do we define it this way, instead of just allowing it simulationist programs to make \emph{any} calls of the form $\apply(p_j,(p_{-j}',(r_m')))$, including where the $p_{-j}'$ are not themselves simulationist? Let's imagine we instead used the latter definition. Consider the Prisoner's Dilemma. We could then have a program $p_1$ that runs a call of the form $\apply(p_2,(p^*,(r_m')))$, where $p^*$ is a program that checks something syntactic about its input program. If $p_2$ were then something like an $\eps$GroundedFairBot, $p_2$ would run $p^*$ on $p_2$, whose output would depend on syntactical details of $p_2$, and hence $p_2$'s output would also depend on syntactical details of $p_2$, which would reveal syntactical information about $p_2$ to $p_1$. $p_1$ does not then seem to fit the spirit of `simulationist'.

One can also consider a 3-player example in which programs $p_1$ and $p_2$ would be `simulationist' under the latter example, but their joint behaviour would not even be behaviourist (i.e., they could give different outputs on two programs that give the same output on every input). To see this, let $p_1$ simulate $p_2$ with input programs $p^*$ and $p_3$, and $p_2$ simulate its first input program with input (its second input program, $p_3^*$) for some $p_3^*$. Then, the simulation of $p_2$ would run $p^*$ on input $p_3$, $p_3^*$. If $p^*$ then is a program that outputs a certain action iff its inputs are exactly equal, this would then tell $p_2$ whether $p_3$ is syntactically equal to $p_3^*$, and this information could be passed to $p_1$ via $p_2$'s action. This would allow the output of $p_1$ on input $p_2,p_3$ to depend purely on syntactical details of $p_3$.

\section{Further results on more general simulationist programs}\label{app:gen_simul_further}
\begin{restatable}{prop}{outputIndepHalting}\label{prop:output_indep_halting}
    Correlated $\eps$Grounded$\pi$Bots halt output independently against each other. The same holds for uncorrelated $\eps$Grounded$\pi$Bots.
\end{restatable}

The below result is an equivalent formulation of \Cref{prop:simulationistNegativeResultGame}. One may think of $\delta_j^t$ below for each player $j$ as roughly the probability of player $j$ being on `time step' $t$ (equivalent to $\prob{T_j+1=t}$ in \Cref{prop:simulationistNegativeResultGame}), $c_j^t$ as the strategy player $j$ plays on `time step' $t$ at equilibrium, $d_i$ as a strategy that player $i$ defects to after `time step' $t_0$, and $s_j$ as the punishment strategy chosen by player $j$ given $t_0$ and $d_i$.
\begin{restatable}{thm}{simulationistNegativeResult}\label{thm:simulationist_negative_result}
Let $p_1,\ldots, p_n$ be simulationist programs for game $G$, in the uncorrelated setting. Assume $p_1, \ldots, p_n$ halt against each other output independently.
Let $c_i=\E{\apply(p_i,(p_{-i},(r_m)))}$ be the strategy of program $p_i$ against programs $p_{-i}$, and $v_i = u_i(c_{1:n})$ be the utility for player $i$.
Suppose that $(p_1,\ldots, p_n)$ is a program equilibrium. Then there exist sequences $(\delta^t_1)_t,\ldots,(\delta^t_n)_t\in [0,1]^\infty$ with, for all $j$, $\sum_{t=1}^\infty \delta_j^t = 1$, and $(c^t_1)_t,\ldots, (c^t_n)_t$ with $c^t_j\in \Delta(\mathcal{A}_j)$, such that $\sum_{t=1}^\infty \delta_j^tc_j^t = c_j$ for each $j$, with:
\begin{align*}
    &v_i \geq \max_{\substack{t_0\in \mathbb{N}\\
    d_i \in \Delta(\mathcal{A}_i)}}\min_{
    s_j\in \Delta(\mathcal{A}_j),\text{ each } j\neq i} u_i\left(\sum_{t=1}^{t_0-1}\delta_i^tc_i^t+d_i\sum_{t\geq t_0}\delta_i^t,\left(\sum_{t=1}^{t_0}\delta_j^t c_j^t + s_j\sum_{t> t_0}\delta_j^t\right)_{j\neq i}\right).
\end{align*}
(Where the above denotes player $j$, for $j\neq i$, playing according to $\sum_{t=1}^{t_0}\delta_j^tc_j^t + s_j\sum_{t>t_0}\delta_j^t$.)

\end{restatable}

In the above, we cannot make the simplifying assumption that all sequences $(\delta_i^t)$ are geometric -- see \Cref{app:cant_assume_geo}.

\begin{defn}
    Say that a strategy profile $(c_1,c_2)$ is \emph{strictly Pareto optimal in the limit of small deviations} if it is strictly Pareto optimal, and for all $(s_1,s_2)\neq (c_1,c_2)$, we have $u_i(s_1,c_2)+u_i(c_1,s_2)<u_i(c_1,c_2)$, for some $i$.
\end{defn}

\Cref{thm:simulationist_negative_result} doesn't rule out strategy profiles which are strictly Pareto optimal in the limit of small deviations and strictly individually rational for two players:

\begin{restatable}{thm}{noRuleOutPareto}
    All strategy profiles that are strictly Pareto optimal in the limit of small deviations and strictly individually rational in two player games satisfy the conditions of \Cref{thm:simulationist_negative_result}/\Cref{prop:simulationistNegativeResultGame}.
\end{restatable}

It remains an open question whether there always exist simulationist programs (e.g., perhaps something like \Cref{app:gen_simul_idea}, \Cref{algorithm:gen_simul}) attaining such strategy profiles at equilibrium.

\section{An idea for more general simulationist programs}\label{app:gen_simul_idea}

Here, we sketch an idea for more general simulationist programs, using the following ideas:
\begin{enumerate}
    \item More general distributions over time step. Previously, we used a geometric distribution over the virtual `time step' on which all players find themselves. We can generalise this by instead having a different random variable $T_j((r_m))$ for each player $j$, as long as program $i$, given $T_i= t$, only ever calls program $j$ with $T_j < t$.
    \item Multiple independent samples at a given time step. Taking purely independent samples causes problems with halting. However, we may take independent samples of player $j$'s action from the distribution $(r_m)\mid T_j = t$, while still guaranteeing that the call tree has maximum depth at most $t$ (assuming the above bullet point holds).
\end{enumerate}

Before we construct this more general class of programs, we will first need to define some things:

First, we will define random variables $T_j$ as follows: For each $j$, let $(\delta_j^t)_{t\geq 0}$ be some sequence in $[0,1]^\infty$ with $\sum_{t}\delta_j^t = 1$. Then we construct $T_j$ such that $\prob{T_j = t} = \delta_j^t$ via $T_j((r_m)_{m\geq 0}) = \max\left\{t: \sum_{t'\leq t}\delta_j^{t'} \leq r_0\right\}$. Now, define functions to sample from $(r_m)\mid T_j = t$ as follows:

\begin{defn}
    Given a sequence $(r_m)$, define $f_{t,j}: [0,1]^\infty \rightarrow [0,1]^\infty$ via $f_{t,j}((r_m))_k = r_k$ for $k\geq 0$ and $f_{t,j}((r_m))_0 = \sum_{t'<t} \delta_j^{t'} + r_0\delta_j^t$.
\end{defn}

Note that $T_j(f_{t,j}(r_m)) = t$ for all $(r_m)$, and if $((r_m))$ follows a $U[0,1]$ i.i.d.\ distribution, $f_{t,j}((r_m))$ is distributed as $(r_m)\mid T_j = t$. Now, we construct a way to obtain independent samples from $((r_m))\mid T_j = t$.

\begin{defn}
    Given a sequence $(r_m)$, for each $k$ and each $i=1,\ldots k$, let $q_{i,k}((r_m)) = (r_{lk+i+1})_{l\geq 0}$, so that $q_{1,k},\ldots q_{k,k}$ partition the sequence $(r_m)_{m\geq 1}$. Then, let $F_{t,j,i,k}((r_m)) = f_{t,j}(q_{i,k}((r_m))$.
\end{defn}

Essentially, $F_{t,j,i,k}$ is the $i$th of $k$ independent samples of sequences $(r_m)$ with $T_j = t$. Note that $F_{t,j,i,k}(F_{t',j',i',k'}((r_m))$ does not depend on $j'$ (or $t'$), because $f_{t',j'}$ only affects the first element of the sequence, which is subsequently truncated by $q_{i,k}$. We are now ready to construct our programs.
\begin{algorithm}[ht]
\caption{More general simulationist program $p^*_i$}
\label{algorithm:gen_simul}
\begin{algorithmic}
    \STATE {\bfseries Input:} All programs' source codes source codes $p_{-i}$,
    \STATE {\phantom{{\bfseries Input:}}} sequence $(r_m)_{m=0}^\infty\in[0,1]^{\omega}$ sampled uniformly
    \STATE $T \leftarrow T_i((r_m))$
    \FOR{$t=T$ to $t_0$ decreasing}
        \STATE{$a_j^{t,1}\leftarrow \apply(p_j,(p_{-j},F_{t,j,1,1}((r_m))))$}
        \STATE{$c_j^{t} \leftarrow \consistent(a_j^{t,1}, p_j^*, F_{t,j,1,1}((r_m)))$}
        \STATE{$\hat{a}_j^t \leftarrow a^{t,1}_j$}
    \ENDFOR
    \FOR{$t= t_0 -1$ to $1$ decreasing}
        \FOR{$l = 1$ to $k$}
            \STATE{$a_j^{t,l}\leftarrow \apply(p_j,(p_{-j},F_{t,j,l,k}((r_m))))$}
            \STATE{$c_j^{t,l} \leftarrow \consistent(a_j^t, p_j^*, F_{t,j,l,k}((r_m)))$}
        \ENDFOR
        \STATE{$c_j^t = \min_l c_j^{t,l}$}
        \STATE{$\hat{a}_j^t \leftarrow \frac{1}{k}\sum_{l=1}^ka^{t,l}_j$}
    \ENDFOR
    
    \RETURN{$\pi_i(\bm{\hat{a}}^{1:T},\bm{c}^{1:T})$}
\end{algorithmic}
\end{algorithm}

In the above, the function $\consistent(a,p^*_j,(x_m))$ works as follows: the program computes all $\apply$ calls that would be made directly by the simulationist program $p_j^*$ on the input sequence $(x_m)$, runs these with any instances of $p_j^*$ replaced by $p_j$, and then checks if the input action $a$ is the same as the output of $p_j^*$ given the results of the $\apply$ calls (and outputs 1 if so and 0 otherwise). The $p_j^*$ we would use would then be another instance of the above program, with the same choices of $T_i$s and the same simulation scheme.

The idea is to take $k$ independent simulations of each program up until time step $t_0$ (or $T$, whichever is sooner), and then only 1 simulation thereafter. To avoid punishing programs for deviating from the equilibrium to punish another program, we check whether program's outputs are consistent with what they should be doing. This scheme then allows the programs to get a better approximation of the action distribution for earlier time steps. The hope is that later deviations can just be punished with whatever action is worst for the deviating player when everyone is playing according to the equilibrium, since such deviations will occur with small probability.

We hypothesize, but do not prove, the following result:

\begin{conj}
    Suppose that an equilibrium with payoffs $v_{1:n}$ may be sustained in the following game: Each player $i$ observes $T_i$, sampled independently from other players, and sees $k$ independent samples of each other player $j$'s policy conditional on $T_j = t$, for all $t<\min\{t_0,T_i\}$, and 1 sample of each other player $j$'s policy conditional on $T_j = t$, for all $t_0\leq t< T_i$, along with observations associated with each sample (i.e., a sample of the policies each sampled player would observe), and then chooses a strategy.
\end{conj}

A further open question is then how large the gap is between this result and \Cref{prop:simulationistNegativeResultGame}.
\section{Cannot simplify sequences in \Cref{thm:simulationist_negative_result}}\label{app:cant_assume_geo}

In \Cref{thm:simulationist_negative_result}, we cannot make the simplifying assumption that the sequences $(\delta_j^t)$ are geometric (i.e., $\delta_j^t = \eps_j (1-\eps_j)^t$ for some $\eps_j \in (0,1]$), as the below example shows:

\begin{eg}
    Consider the following program game between Player 1 `the Government' and Players $2,\ldots 100$, `the Citizens':
    \begin{itemize}
        \item The Government (Player 1) chooses an action from the set \{Liberal democracy ($L$), Dictatorship ($D$)\}.
        \item The Citizens (Players 2 to 100), whom the game is symmetric between, each choose an action from the set \{Pay taxes ($P$), Evade taxes ($E$)\}.
        \item Utilities for the Government depend on the number of players that play $E$, as follows:
        \begin{itemize}
            \item If no players play $E$, the Government receives utility $u_1(L, 0) = 0$ if it plays $L$, or $u_1(D,0) = 1$, if it plays $D$.
            \item If $1\leq k < 99$ players play $E$, the Government receives utility $u_1(L,k) = -10$, $u_1(D,k) = 1$. (We can imagine that Citizens not paying their taxes costs the Government money, but this cost is negated if the government is a dictatorship, since their neighbours snoop on them and report them and then they are fined a lot of money.)
            \item If all 99 Citizens play $E$, the Government has no money and is overthrown: $u_1(L,99) = u_1(D, 99) = -100$.
            \item In words, playing $D$ gains the Government 1 utility, unless all Citizens evade taxes, which costs the government 100 regardless of its action. Some but not all citizens evading taxes costs the government 10 if it plays $L$, and nothing if it plays $D$.
        \end{itemize}
        \item Utilities for each Citizen is a function only of its own action, and that of the government, as follows:
        \begin{itemize}
            \item If Player $j$ pays taxes, and the government is a democracy, the player receives utility 0: $u_j(L,P) = 0$. The player then gains utility 1 by evading taxes: $u_j(L,E) = 1$.
            \item If the government is a dictatorship, and Player $j$ pays taxes, Player $j$ loses $0.1$ utility: $u_j(D,P) = -0.1$. If the government is a dictatorship, and Player $j$ evades taxes, he loses $10$ utility: $u_j(D,E) = -10$.
        \end{itemize}
    \end{itemize}
    Can the $(L,P,\ldots, P)$ action profile be attained in a simulationist program equilibrium? We will show that \Cref{thm:simulationist_negative_result}/\Cref{prop:simulationistNegativeResultGame} shows that this is impossible with geometric $(\delta_j^t)_t)$, but does not rule out the possibility of attaining it with more general sequences.

    First, we show that there is no $(L,P,\ldots,P)$ equilibrium with geometric sequences $(\delta_j^t)_t$. Suppose that for all $j = 1,\ldots, n$ and $t = 0,1,2,\ldots$, we have $\delta_j^t = \eps_j(1-\eps_j)^{t}$, for some $\eps_j \in (0,1]$. Let $\Lambda_j = 1-\eps_j$, so that $\Lambda_j^t = \sum_{t'\geq t} \delta_j^t$. We may view $\Lambda_j^t$ as the probability of player $j$ being on time step at least $t$, and `seeing' the distribution of what other players would do at earlier time steps.

    Now, we will argue that in order for this to be an equilibrium, we need $\Lambda_1^{t+1} > \Lambda_j^t$ for $j\neq 1$ and $t$ sufficiently large (and hence $\Lambda_1 > \max_{j\neq 1} \Lambda_j$), but $\Lambda_1^t < \max_{j\neq 1} \Lambda_j^{t+1}$ for smaller values of $t$ (and hence $\Lambda_1 < \max_{j\neq 1} \Lambda_j$), a contradiction.

    For the first part, consider Citizen $j$ deviating to play $E$ at time steps $t\geq t_0$. She then plays $E$ with overall probability $p = \Lambda_j^{t_0}$, and is detected by the Government with overall probability $q = \Lambda_1^{t_0+1}$. She then gains expected utility (compared to just playing $P$) at least:
    \begin{align*}
        p(1-q) - 100pq  -0.1(1-p)q &= p -100.9pq -0.1q\\
        &= \Lambda_j^{t_0}-100.9\Lambda_j^{t_0}\Lambda_1^{t_0+1} - 0.1\Lambda_1^{t_0+1}\\
        &\sim \Lambda_j^{t_0}-0.1\Lambda_1^{t_0+1} \qquad \text{ as }t_0 \rightarrow \infty.
    \end{align*}
    Thus, in order for this to be at most 0 for all $t_0$, we will need $10\Lambda_j^{t_0}<\Lambda_1^{t_0+1}$ in the limit $t_0\rightarrow \infty$, which requires $\Lambda_1 > \Lambda_j$. This holds for all $j>1$.

    For the second part, consider the Government deviating to play $D$ at time steps $t\geq t_0$, with overall probability $p=\Lambda_1^{t_0}$, and being detected by player $j$ with probability $q_j = \Lambda_j^{t_0+1}$. Let $q = \max_j q_j = \max_j \Lambda_j^{t_0+1} < \Lambda_1^{t_0+1}<\Lambda_j^{t_0} = p$, and note that it is always worse for the government for more Citizens to play $E$, and so the government gains, by deviating, utility at least: 
    \begin{align*}
        &p(1-q^{99})-10(1-p)(1-q^{99}-(1-q)^{99}) - 100q^{99}\\
        &\geq p(1-q^{99})-10(1-p) - 100q^{99}\\
        &\geq p - 10(1-p) -100q^{99}
    \end{align*}
    Now, let $t_0$ be maximal such that $p\geq 0.93$. Then we must have that $q<\lambda_1^{t_0+1}<0.93$, and so the gain from deviating is at least $0.93 - 10\times 0.07 - 101 \times 0.93^{99} \approx 0.23 > 0$. Hence, there is no simulationist equilibrium with these payoffs with geometric sequences for the $(\delta_j^t)$s.

    However, \Cref{thm:simulationist_negative_result} does not rule out the possibility of an $(L,P,\ldots,P)$ equilibrium in general. To see this, we will give sequences $(\delta_j^t)$ for which, for all $i$, there is no choice of $d_i$ and $t_0$ such that player $i$ can profitably deviate to playing $d_i$ starting on time-step $t_0$.

    Let 
    \begin{equation*}
        \delta_1^t = \begin{cases*}
            \frac{4}{5} & if $t=0$\\
            \frac{1}{5}\times 0.01\times 0.99^{t-1} & for $t\geq 1$
        \end{cases*}
    \end{equation*}
    and for $j \neq 1$:
    \begin{equation*}
        \delta_j^t = \begin{cases*}
            0 & if $t=0$\\
            \frac{99}{100} & if $t=1$\\
            \frac{1}{100}\times 0.01\times 0.99^{t-2} & for $t\geq 2$
        \end{cases*}
    \end{equation*}
    First, we argue that player 1 cannot profitably deviate. Suppose that Player 1 deviates to playing $D$ at time step $t_0$ (noting that there is no benefit to instead playing a mixed strategy, since the detection probability then remains the same, and $D$ weakly dominates $L$). We now have deviation probability $p = \sum_{t\geq t_0}\delta_1^t$, and probability $q= \sum_{t>t_0} \delta_2^t$ of each of the other players detecting the deviation. The gain from deviating is then $p(1-q^{99})-10(1-p)(1-q^{99}-(1-q)^{99})-100q^{99}$.

    Let's consider the cases. If $t_0 = 0$, we have $p=1$, $q=1$, and the gain from deviating is $-100$. If instead $t_0\geq 1$, we have $p = \frac{1}{5}\sum_{t\geq t_0}0.01\times 0.99^{t-1} = \frac{1}{5}\times 0.99^{t_0-1}$ and $q = \sum_{t\geq t_0+1}\frac{1}{100}\times0.01\times0.99^{t-2} = \frac{1}{100}\times 0.99^{t_0-1}$.

    Thus, $p = 20q$ and the gain from deviating is at most
    \begin{align*}
        &p(1-q^{99})-10(1-p)(1-q^{99}-(1-q)^{99})-100q^{99}\\
        &\leq p -10(1-p)(1-(1-q)^{99}) +10(1-p)q^{99}-100q^{99}\\
        &\leq p -10(1-p)(1-(1-q)^{99}) -90q^{99}\\
        &\leq p -10(1-p)(1-(1-q)^{99})\\
        &= 20q - 10(1-20q)(1-(1-q)^{99}).
    \end{align*}
    Now, we have $(1-q)^{-99} \geq 1+99q$, and so $(1-q)^{99}\leq \frac{1}{1+99q}$. Thus, the gain is at most
    \begin{equation*}
        20q - 10(1-20q)\frac{99q}{1+99q}
        = 10q\left(2-\frac{99(1-20q)}{1+99q}\right)
    \end{equation*}.
    Since $q\leq \frac{1}{100}$, this is in turn at most $10q(2-99\times\frac{2}{5})<0$.

    Can the citizens profitably deviate? Suppose Citizen $j$ begins deviating at time step $t_0$.

    With deviation probability $p$, and detection probability $q$ the gain from deviating is then at most 
    \begin{equation*}
        p(1-q)-100pq - 0.1(1-p)q = p -101pq -0.1q +0.1pq \leq p -0.1q -100pq.
    \end{equation*}
    If $t_0 \leq 1$, we then have deviation probability $p \leq 1$, and detection probability $q \geq 0.99 \times \frac{1}{5}$. Thus, the gain from deviating is at most $p(1-100q) \leq p(1-\frac{99}{5})<0$.

    If instead $t_0\geq 2$, we have $p \leq \frac{1}{100}\times 0.99^{t_0-2}$ and $q= \frac{1}{5}\times 0.99^{t_0}$. Thus, the gain from deviating is at most $p-0.1q$, which is in turn at most $p-0.99^2\times 20\times 0.1 p < 0$.

\end{eg}

\section{Proofs}\label{app:proofs}

\subsection{Proofs for Section \ref{sec:correl}}\label{app:correl_proofs}
In order to prove our halting results, we will first give more general conditions under which halting is guaranteed.

\begin{lemma}\label{lemma:decreasingT}
Let $p_1,\ldots p_n$ be programs for an $n$-player correlated program game satisfying the following conditions:
 \begin{enumerate}[label = \textbf{(C\arabic*)}]
    \item \textbf{Probability zero events occur in apply calls with probability zero:} For all $i$, for any set $A$ with $\prob{((r_m),(x_m))\in A} = 0$, $p_i$ almost surely makes no standard apply (or $\apply^*$) calls with input sequences in $A$.\label{cond:zero_measure_apply_calls}
    \item \textbf{Halts when standard apply calls halt:} For all $i$, for any $p_{-i}$, given that all of its standard apply/$\apply^*$ calls halt, $p_i$ halts almost surely. (More precisely, the event that all $\apply$/$\apply^*$ calls halt and $p_i$ does not halt has probability zero.) \label{cond:halting}
    \item \textbf{Only finitely many direct apply calls:} For all $i$, the number of standard apply/$\apply^*$ calls made directly by $p_i$ is almost surely finite. (By `directly', we exclude here $\apply$/$\apply^*$ calls made within other $\apply$/$\apply^*$ calls.)\label{cond:finite}
    \end{enumerate}

Finally, suppose that there exist random variables $T(i,(r_m))$ taking values in $\mathbb{N}\cup \{\infty\}$, for each $i= 1, \ldots n$, almost surely finite, such that, for all $i$, given input $p_{-i}$ and sequence $(r'_m)$, all standard apply/$\apply^*$ calls made by $p_i$ of the form $\apply(p_j,(p_{-j},(r''_m), \cdots)$ have $T(j,(r''_m))<T(i,(r'_m))$ (almost surely).

Then for all $i$, $p_i$ almost surely halts on input $(p_{-i}, (r_m),(x_m))$.
\end{lemma}

\begin{proof}[Proof of \Cref{lemma:decreasingT}]
We proceed by essentially showing halting conditional on $T(i,(r_m)) = t$, for each positive probability such $t$. Since $T(i, (r_m))$ is almost surely finite for all $i$, it suffices to show that for any finite $t$, we have that the probability of the event $\{((r_m),(x_m)):T(i,(r_m))=t$ and $p_i$ fails to halt on input $(p_{-i},(r_m),(x_m))\}$ is 0, for all $i$.
We proceed by induction on $t$.

\bm{$t=0$}: If $T(i,(r_m))=0$, we have that, for all $i$, $p_i$ almost surely makes no standard $\apply$/$\apply^*$ calls (since those would almost surely have to have $T<0$, which is impossible). Hence, by \ref{cond:halting}, we have that the event \{$p_i$ does not halt and $T(i,(r_m))=0$\} occurs with probability 0. Either way, we are done. 

\bm{$t>0$}: Assume that for each $k=0,\ldots,t-1$, and all $i$, we have that either $\prob{T(i,(r_m))=k}=0$ or $p_i$ almost surely halts conditional on $T(i,(r_m))=k$. Consider any $i$. If $\prob{T(i,(r_m))=t}=0$, we are done, so assume $\prob{T(i,(r_m))=t}>0$. Now, we have that all standard $\apply$ (or $\apply^*$) calls $\apply(p_j,(p_{-j},(r'_m),(x_m')))$ made by $p_i$ almost surely have $T(j,(r'_m))\leq t-1$. Now, by the induction assumption, for all $j$, the event $A = \{((r_m'),(x_m')):T(j,(r'_m)) \leq t-1 \text{ and } p_j \text{ fails to halt on input } (p_{-j},(r'_m),(x'_m))\}$ has probability 0. Hence, by \ref{cond:zero_measure_apply_calls} we must have that $p_i$ almost surely makes no standard apply calls to $p_j$ with input sequences in $A$. The same holds for all $j$. Thus, $p_i$ almost surely makes no standard $\apply$/$\apply^*$ calls that do not halt, and so by conditions \ref{cond:halting} and \ref{cond:finite}, $p_i$ almost surely halts. By induction, we are then done.
\end{proof}

\begin{restatable}{lemma}{applystarscreening}\label{lemma:apply_star_screening}
    For any programs $p_1,\ldots,p_n$, if $\apply^*(p_i,(p_{-i},(r_m),(x_m)))$ halts, its output does not depend on the sequence $(x_m)$.
\end{restatable}
\begin{proof}
    For the $\apply^*$ call to halt, it must make finitely many $\apply^*$ calls. Consider then the $\apply^*$ calls of maximum depth in the tree of $\apply^*$ calls -- i.e., those that do not make further $\apply^*$ calls (they may make further $\apply$ calls, of course). The output of these cannot depend on the $(x_m)$, since that would require the $(x_m)$ to be accessed, which would result in the output being $R_j$ for some $j$, regardless of the $(x_m)$. Meanwhile, if the output of all $\apply^*$ calls directly made by an $\apply^*$ call do not depend on the $(x_m)$, then that $\apply^*$ call cannot depend on the $(x_m)$ either -- that would require accessing the $(x_m)$ outside of an $\apply^*$ call, and so the output would again be $R_j$ for some $j$, regardless of the $(x_m)$. By induction on the tree of $\apply^*$ calls, we are done.
\end{proof}
\correlhalting*
\begin{proof}\label{pf:correlhalting}
We proceed by verifying the conditions of \Cref{lemma:decreasingT}. Let $\eps = \min_i \eps_i$ and $T(i,(r_m)) = \min\{t:r_t<\eps\}$ for all $i$. That is, the time step $T$ as determined by the program with the smallest $\eps$.

\ref{cond:zero_measure_apply_calls}: Let $A$ be an event with probability 0. We have that $p_i$'s $\apply$ calls all have input sequences of the form $(r_m)_{m\geq t}, (x_m)_{m\geq 0}$ for some $t\geq 1$. Thus, the probability that an $\apply$ call has input sequences in $A$ is at most the probability that some such sequence is in $A$. This is in turn bounded above by $\sum_{t=0}^\infty \prob{((r_m)_{m\geq t}, (x_m)_{m\geq 0}) \in A} = \sum_{t=0}^\infty \prob{((r_m)_{m\geq 0}, (x_m)_{m\geq 0}) \in A} = 0$, when the penultimate inequality follows from $(r_m)_{m\geq t}$ having the same distribution for all $t$.

\ref{cond:halting}: Since $\pi_i$ is by assumption computable, $p_i$ halts whenever its standard $\apply^*$/$\apply$ calls do (noting that all $p_i$'s $\apply^*$/$\apply$ calls are standard).

\ref{cond:finite} $p_i$ makes at most $nT_i$ $\apply^*/\apply$ calls, where $T_i$ is the time step as determined by program $i$, i.e. $\min\{m:r_m<\eps_i\}$. Since $\eps_i>0$, we have $T_i$ follows a $\Geo(\eps_i)$ distribution, so $nT_i$ almost surely finite.

We then verify the conditions concerning $T= T(i,(r_m))$. Since $\eps>0$, we have $T$ is a geometric random variable and so is almost surely finite. It remains to check that $T$ is strictly decreasing under $\apply^*$ calls. We have that program $p_i$ makes only $\apply^*$ calls with input sequences $(r_m)_{m\geq 1}, \ldots (r_m)_{m\geq T_i}$, where $T_i$ is as above. Now, $T\geq T_i$ by our choice of $\eps$. Hence, for $1\leq t \leq T_i$, we have $T((r_m)_{m\geq t}) = T((r_m)_{m\geq 0}) - t \leq T((r_m)_{m\geq 0}) - 1$, so we are done.
\end{proof}

\correlHaltingExt*
\begin{proof}\label{pf:correlHaltingExt}
    We again verify the conditions of \Cref{lemma:decreasingT}. By the proof of \Cref{thm:correl_halting}, we already know that programs $p_{-i}$ satisfy conditions \ref{cond:zero_measure_apply_calls}, \ref{cond:halting} and \ref{cond:finite}, whereas $p_i$ satisfies these by assumption. We now proceed by constructing $T$ which is almost surely finite and decreasing under the standard $\apply$/$\apply^*$ calls made by the programs. First, let $\eps = \min_{j\neq i}\eps_j$ and  $D((r_m)) = \min\{m:r_m,\ldots r_{m+c}<\eps\}$. First, we will show that $D$ is almost surely finite: For $k = 0, 1, 2, \ldots$, let $A_k$ be the event that $\{r_{(c+1)k},\ldots r_{(c+1)k+c}<\eps\}$. We have that events $A_0,A_1,\ldots$ are independent, and all have probability $\eps^{c+1}>0$. Thus, the probability that at least one of $A_0,\ldots, A_k$ occurs is equal to $1-(1-\eps^{c+1})^k\rightarrow 1$ as $k\rightarrow \infty$. Meanwhile, if at least one of $A_0,\ldots, A_k$ occurs, we have $D((r_m))\leq (c+1)k$. Thus, $D((r_m))$ is almost surely finite. Then, define $T(j,(r_m))$ as follows:

    \begin{equation*}
    T(j,(r_m)) \defeq \begin{cases*}
        0 &if $j\neq i$  and $r_0<\eps$\\
        1 &if $j=i$ and $r_0,\ldots, r_c<\eps$\\
        2D((r_m)) &if $j\neq i$ and $r_0\geq \eps$\\
        1+2D((r_m)) &if $j=i$ and one of $r_0,\ldots, r_c$ is at least $\eps$\\
    \end{cases*}
    \end{equation*}
Since $D((r_m))$ is almost surely finite, $T(j,(r_m))$ is almost surely finite for all $j$. It remains to show that $T$ is strictly decreasing under standard $\apply$/$\apply^*$ calls. I.e., that for all $j$ and all input sequences $(r_m')$, if $p_j$, given input programs $p_{-j}$ and shared sequence $(r_m')$ makes an $\apply^*$ (or $\apply$) call of the form $\apply^*(p_k,(p_{-k},(r''_m),\cdots))$, we must have $T(k,(r''_m))<T(j,(r_m'))$. We consider the four cases in turn:

\begin{enumerate}
    \item If $j\neq i$, and $r_0<\eps$, we have that $p_j$ makes no standard $\apply$/$\apply^*$ calls. Hence, the condition is trivially satisfied.
    \item If $j=i$, and $r_0,\ldots,r_c < \eps$, we have that $p_i$'s standard $\apply$/$\apply^*$ calls all have input of the form $(p_k,(p_{-k},(r'_m),(x'_m)))$, where $(r'_m) = (r_m)_{m\geq t}$ for some $0\leq t \leq c$ and $k\neq i$. Then $r'_0 = r_t <\eps$, and so $T(k,(r'_m))=0<1 = T(j,(r_m))$.
    \item If $j\neq i$ and $r_0\geq \eps$, we have $T(j,(r_m)) = 2D((r_m))\geq 2$. Let $T_j = \min\{t:r_t<\eps_i\}$. Then all standard $\apply$/$\apply^*$ calls have input sequence of the form $(r_m)_{m\geq t}$ for $1\leq t \leq T_j$. Meanwhile, we must have $D((r_m)) \geq \min\{t:r_t<\eps\} = \max_{j\neq i} T_j$. Hence, for $t \in \{1,\ldots,k\}$ we have $D((r_m)_{m\geq t}) =D((r_m)_{m\geq 0}) - t \leq D((r_m)_{m\geq 0}) - 1$. So if $k\neq i$, $\apply^*$ calls to $p_k$ either have $T=0$ or $T = 2D((r_m)_{m\geq t}) \leq 2D((r_m)) - 2 < T(j,(r_m))$. Finally, if $k = i$, we either have $T(i,(r_m)_{m\geq t}) = 1 <2$ or $T(i,(r_m)_{m\geq t}) = 1 + 2D((r_m)_{m\geq t}) \leq 2D((r_m))-1 < T(j,(r_m))$.
    \item Finally, if $j=i$ and one of $r_0,\ldots r_c$ is at least $\eps$, we have that $p_i$ makes $\apply^*$/$\apply$ calls with input of the form $(p_k,(p_{-k},(r_m)_{m\geq t}))$ where $0\leq 0 \leq c$. If $t \leq D((r_m))$, we then have $D((r_m)_{m\geq t}) = D((r_m))-t \leq D((r_m))$ and so $T(k,(r_m)_{m\geq t}) \leq 2D((r_m))<T(j,(r_m))$. Otherwise if $D((r_m)) = d < t$, we have $r_d,\ldots r_{d+c} < \eps$ and $d+c\geq c \geq t$ and so $r_t<\eps$ and $T(k,(r_m)_{m\geq t}) = 0$.
\end{enumerate}
\end{proof}

\expectedRuntimeCorrel*
\begin{proof}\label{pf:expectedRuntimeCorrel}
    For the first part, let $\eps_j\leq \eps_i<1/2$. Let $f_i(t)$ be the expected total number of $\apply^*$ calls (including those made within $\apply^*$ calls) made by $p_i$ conditional on $T_i=t$, where $T_i = \min\{t:r_t<\eps_i\}$. Then we have $f_i(1) = n\geq 2$, and for $t>1$, $f_i(t) \geq f_i(t-1)+f_j(t-1)$, noting that that $p_i$ makes $\apply^*$ calls including $\apply^*(p_i,(r_m)_{m\geq 1})$ and $\apply^*(p_j,(r_m)_{m\geq 1})$. The former has $T_i=t-1$, and the latter has $T_j\geq t-1$, since $\eps_i\leq \eps_j$. Meanwhile, on any given input, the $\apply$ calls made by $p_j$ are a superset of those made by $p_i$, and so we have, for all $t$, $f_j(t)\geq f_i(t)$. Hence, $f_i(t)\geq 2f_i(t-1)$. By induction, we obtain $f_i(t) \geq 2^t$ for $t\geq 1$. But $\E{2^T} = \sum_{t=1}^\infty \eps_i(1-\eps_i)^t 2^t=\infty$. Thus, $p_i$ has infinite expected runtime. Now, all other programs $p_k$ call $p_i$ with probability $1-\eps_k$, with input sequence $(r_m)_{m\geq 1}$, which follows an i.i.d.\ $U[0,1]$ distribution (conditional on $r_0>\eps_k$). Hence, all programs $p_k$ with $\eps_k < 1$ have infinite expected runtime.
    
    For the second part, let $\eps= \min_i\eps_i$. Let $T((r_m)) = \min\{m\geq 0: r_m<\eps\}$, and $T_i((r_m)) = \min\{m\geq 0: r_m<\eps_i\}$. Note that $T_i \leq T$ for all $i$ and $(r_m)$. Then, all programs halt immediately (i.e., without simulating the other programs) on input $(r_m)_{m\geq T}$ (regardless of $(x_m)$), and given input $(r_m)_{m\geq t}$ with $t\leq T$, no calls with first input sequence of the form $(r_m)_{m\geq t'}$ with $t'>T$ are made. Thus, overall, no $\apply^*$ calls with first input sequence of the form $(r_m)_{m\geq t}$ with $t>T$ are made. Hence, apply calls all have first input sequence of the form $(r_m)_{m\geq t}$ for $t\in \{1,\ldots, T\}$ and so at most $nT$ distinct $\apply^*$ calls are ever made.

    Now, assume that for all $i$, $\pi_i(\bm{a}^{1:t})$ has complexity $\leq \lambda (t+1)^k$, for some fixed $k$ and $\lambda$. Then the deepest possible $\apply^*$ calls, those with input sequence $(r_m)_{m\geq T}$, each have runtime at most $\lambda$. Given that we have already computed $\bm{a}^{1:t-1}$, computing $a_j^t =\apply^*(p_j,(p_{-j},(r_m)_{m\geq T+1-t}))$ then adds additional runtime at most $\lambda t^k$. Thus, total runtime is at most $n\sum_{t=0}^T\lambda (t+1)^k<n\lambda \sum_{t=0}^T (T+1)^k = n\lambda (T+1)^{k+1}$. Now, $T+1$ is geometric with $\prob{T+1 = t} = \eps(1-\eps)^{t-1}$ and so $\E{(T+1)^{k+1}} = \sum_{t=1}^\infty \eps (1-\eps)^{t-1} t^{k+1}$. Meanwhile, by considering the power series for $(1-\eps)^{-t/2}= \exp(\frac{-\log(1-\eps)t}{2})$, we have
    \begin{equation*}
        (1-\eps)^{-t/2}  \geq\frac{t^{k+1} (-\log(1-\eps))^{k+1}}{2^{k+1}(k+1)!}
    \end{equation*}
    and so
    \begin{align*}
        \E{(T+1)^{k+1}} = &\sum_{t=1}^\infty \eps (1-\eps)^{t-1} t^{k+1}\\
        &\leq \sum_{t=1}^\infty \eps (1-\eps)^{t/2-1} t^{k+1}\frac{2^{k+1}(k+1)!}{t^{k+1}(-\log(1-\eps))^{k+1}}\\
        &=\frac{\eps 2^{k+1}(k+1)!}{(1-\eps)^{1/2}(-\log(1-\eps))^{k+1}}\sum_{t=0}^\infty(1-\eps)^{t/2}\\
        &= \frac{\eps 2^{k+1}(k+1)!}{(1-\eps)^{1/2}(-\log(1-\eps))^{k+1}(1-(1-\eps)^{1/2})}\\
        &\sim \frac{\eps 2^{k+1}(k+1)!}{(1-\eps)^{1/2}\eps^{k+1}(1-(1-\eps)^{1/2})}\\
        &\sim \frac{\eps 2^{k+1}(k+1)!}{\eps^{k+1}\eps/2}\\
        &= O(\eps^{-(k+1)})
    \end{align*}
    as $\eps\rightarrow 0$. Hence, the expected runtime is finite, and polynomial in $\eps^{-1}$.
 \end{proof}
 \begin{lemma}[\citet{actionSequence}, Lemma 2]\label{lemma:action_sequences}
     Suppose that payoffs $v_{1:n}$ for game $G$ are feasible with correlation and strictly individually rational. Then for all $\delta> 0$, for $\eps>0$ (depending on $\delta$) sufficiently small, there is a deterministic sequence of pure strategies $(\alpha_{1:n}^t)_t$ with discounted average payoffs $v_{1:n} = \sum_{t=1}^\infty \eps (1-\eps)^{t-1} u_{1:n}(\alpha_{1:n}^t)$, such that for all $t$, we have that the continuation payoffs $\sum_{k=t}^\infty \eps (1-\eps)^{k-t} u_{1:n}(\alpha_{1:n}^k)$ are within $\delta$ of $v_{1:n}$.
 \end{lemma}
 \repFolkThm*
 \begin{proof}\label{pf:repFolkThm}
    Let $v_{1:n}$ be payoffs that are feasible with correlation and strictly individually rational. Let $v_i^*$ be the minimax payoff of player $i$, for each $i$, and let $\delta>0$ be such that $v_i> v_i^*+2\delta$ for each $i$. By \Cref{lemma:action_sequences}, we may choose a sequence of action profiles $(\alpha_{1:n}^t)_{t=1}^\infty$ such that $v_{1:n} = \sum_{t=1}^\infty \eps (1-\eps)^{t-1} u_{1:n}(\alpha_{1:n}^t)$ and for all $t$ and $i$, we have $\sum_{k=t}^\infty \eps (1-\eps)^{k-t} u_i(\alpha_{1:n}^k) > v_i-\delta > v_i^* + \delta$.

    Then, construct policy profiles $\pi_{1:n}, \pi_{1:n}^*$ for the repeated game with correlation and screening inductively as follows:
    Let $\pi_i(\omega;q) = \pi_i^*(\omega;q)= \alpha_i^1$ for all $i$ (where $\omega$ is the empty history). Then, for $t\geq 1$, given a screened history $\bm{a}^{1:t}$, let $\pi_i(\bm{a}^{1:t};q) = \pi_i^*(\bm{a}^{1:t};q)= \alpha_i^{t+1}$ if $\bm{a}^{1:t}=\bm{\alpha}^{1:t}$. Otherwise, if $j$ is the first player in the history to have deviated from this sequence, let $\pi_i(\bm{a}^{1:t};q) = s_{i,j}$, where $s_{-j,j}$ is an independent minimax strategy profile against player $j$ (i.e., $\max_{a_j} u_i(s_{-j},a_j)$ is minimised at $s_{-j} = s_{-j,j}$)). In this case, let $\pi_i^*(\bm{a}^{1:t};q) = R_i$.

    Note that policies $\pi_{1:n}^*$ and $\pi_{1:n}$ are then compatible screened and unscreened policies for the players. If all players play according to $\pi_{1:n}$, we can see that players will play according to $\alpha_{1:n}^t$ on turn $t$, and hence obtain overall expected utilities $\sum_{t=1}^\infty (1-\eps)^{t-1}u_{1:n}(\alpha_{1:n}^t) = v_{1:n}/\eps$. Meanwhile, if player $j$ deviates on turn $t$, she will obtain expected utility at most 
    \begin{align*}
        &\sum_{k < t} u_j(\alpha_{1:n}^k) (1-\eps)^{k-1} + (1-\eps)^{t-1} \max_{a_j}\E{u_j(a_j,\alpha_{-j}^t} + \frac{(1-\eps)^t}{\eps}v_j^*\\
        &\leq v_j/\eps + (1-\eps)^{t-1}\max_{a_{1:n}}u_j(a_{1:n}) + \frac{(1-\eps)^t}{\eps}v_j^* - \sum_{k>t}(1-\eps)^{k-1}u_j(\alpha_{1:n}^k)\\
        &\leq v_j/\eps + (1-\eps)^{t-1}\max_{a_{1:n}}u_j(a_{1:n}) + \frac{(1-\eps)^t}{\eps}v_j^* - \frac{(1-\eps)^t}{\eps}\sum_{k>t}(v_j^* + \delta)\\
        &\leq v_j/\eps - \frac{(1-\eps)^t}{\eps}\left(\delta -\frac{\eps}{1-\eps}(\max_{a_{1:n}}u_j(a_{1:n})-v_j)\right).
    \end{align*}
    where $v_j^*$ is player $j$'s minimax utility. For $\eps$ sufficiently small, this is then less than $v_j$ for all $t$ and all $j$. We are then done.

    Note that we were able to define $\pi_{1:n}^*$ to have screening whenever players randomise independently of the $q$s.

    Moreover, note that although we used grim trigger strategies in this proof, we could instead do something like punish deviating players until they are worse off than if they hadn't deviated, by making maximally pessimistic assumptions about the strategies played by deviating players when there is screening.

    Finally, we could have instead used the shared random variable to randomise instead of a deterministic sequence, if the players had access to past shared random variables in order to determine if the other players deviated. Programs will in fact have access to their analogue of these, and so such a policy would be viable in the program game, although we do not show this here.
 \end{proof}
 \eGpBPolicyCorrel*
 \begin{proof}\label{pf:eGpBPolicyCorrel}
    Follows straightforwardly from the definitions.
\end{proof}
 \correlUtils*
 \begin{proof}\label{pf:correlUtils}
     Let $\bm{a}^{1:t}$ and $T$ be as in \Cref{algorithm:pi-bot-correl}. By \Cref{lemma:screened_hist_dist}, we have that $\bm{a}^{1:t}$ has, for all $t$, given $T\geq t$, the same distribution as the length $t$ screened history of the repeated game $G_\eps$ in which players follow their associated policies. Moreover, $\bm{a}^{1:t}$ does not depend on the $(x_m)$, and so is the same for all programs.

     We will show that conditional on $T=t$, and on $\bm{a}^{1:t}=\bm{\alpha}^{1:t}$ (where this occurs with non-zero probability), the output of the $n$ programs has the same joint distribution as the $\tau_j(\bm{\alpha}^{1:t};q)$, where $q\sim U[0,1]$. Using the above, we will then have that conditional on $T=t$, the output of the programs follows the same joint distribution as the actions of the associated policies on the $t+1$th turn of $G_\eps$. Since the distribution of $T+1$ is the same as the distribution of the number of turns in $G_\eps$, we may then deduce that the expected utility for the program game is equal to the expected utility for the last turn of the repeated game.

     Now, conditional on $T=t>0$, $\bm{a}^{1:t}=\bm{\alpha}^{1:t}$, and $r_{0}=r_{0}'$, the output of program $j$ is distributed as
     \begin{equation*}
         \E{\apply(p_j,(p_{-j},(r_m),(s_m)))\mid T=t, \bm{a}^{1:t}=\bm{\alpha}^{1:t}, r_{0} = r_{0}'} = \tau_j(\bm{\alpha}^{1:t};(r_0'-\eps)/(1-\eps))
     \end{equation*}
     which is the desired distribution, since $(r_0-\eps)/(1-\eps) \sim U[0,1]$, conditional on $T=t$. For $t=0$, the distribution is instead $\tau(\omega;r_0/\eps)$, which is again the desired distribution. It remains then to check that, conditional on $T$, $\bm{a}^{1:T}$, and $r_{0}$, the outputs of the programs are independent. Note that for $j\neq i$, the output of program $j$ is given by $\pi_j(\bm{a}^{1:T};r_{0},x_0^j)$. Thus, conditional on $\bm{a}^{1:T}$ and $r_{0}$, this depends only on $x_0^j$, which is independent between programs. Since this holds for all but at most one program, the outputs of programs are independent given $r_{0}$, $T$, $\bm{a}^{1:T}$. We therefore conclude that the outputs of programs are distributed as the actions on the last turn of $G_\eps$ between players with policies $\tau_j$ and screened policies $\tau_j^*$.

     Since the expected utility on the final turn of a repeated game is $\eps$ times the total expected utility, we are then done.
\end{proof}
\begin{lemma}\label{lemma:screened_hist_dist}
    Consider the setting of \Cref{thm:correl_utils}. Let $\bm{a}^{1:t}$ and $T$ be as in \Cref{algorithm:pi-bot-correl}. Then conditional on $T\geq t\geq 1$ we have $\bm{a}^{1:t}$ follows the same distribution as the screened history in the associated repeated game $G_\eps$ with associated policies, independent of $T$, and of the $(x_m)$.
\end{lemma}
\begin{proof}
First note that since $\bm{a}^{1:t}$ is the output of $\apply^*$ calls, \Cref{lemma:apply_star_screening} gives us that it must be independent of the $(x_m)$.

Now, note that, given $T\geq t\geq 1$, $\bm{a}^{1:t}$ depends only on $(r_m)_{m\geq_{T+1-t}}$. Moreover, $r_{T+1-t:T-1}\simiid U[\eps,1]$ ($t-1$ random variables),  $r_T\sim U[0,\eps)$ and $(r_m)_{m\geq T+1}\simiid U[0,1]$, all independently, conditional on $T\geq t$, independent of $T$. Hence, the distribution of $\bm{a}^{1:t}$ is independent of $T$, conditional on $T\geq t$. Moreover, $\bm{a}^{1:t}$ has the same distribution as if we had $T= t-1$ and $a^{t'}_j = \apply^*(p_j,(p_{-j},(r_m)_{m\geq t-t'}, (x_m)))$.

First, consider $t=1$. Conditional on $r_T=r\in[0,\eps)$, we have  $a_j^1 \sim  \E{\apply^*(p_j,(p_{-j},(r_m)_{m\geq T}, (x_m)_{m\geq 0}))\mid r_T=r}$, which follows the same distribution as $\E{\apply^*(p_j,(p_{-j},(r_m)_{m\geq 0}, (x_m)_{m\geq 0}))\mid r_0=r}$, i.e. as $\tau_j^*(\omega;r/\eps)$. Since $r/\eps$ follows a $U[0,1)$ distribution, it remains to show that, given $r_T$, the $a^1_j$ are independent of each other. For this, we need only note that for $j\neq i$, $p_j$ is an $\eps$Grounded$\pi$Bot and so by \Cref{lemma:eGpB_policy_correl}, $\tau_j^*(\omega;r/\eps)$ is deterministic given $r_0 = r$ (and hence so is $a^1_j$ given $r_T$). Hence, the $a^1_j$ are trivially independent of each other, given $r_T$.

It now suffices to show that, for $t>1$, given $\bm{a}^{1:t-1}$, the $a_j^{t}$ have the same joint distribution as the $\tau_j^*(\bm{a}^{1:t-1};q)$, where $q \simiid U[0,1]$. Conditional on $r_{T+1-t} = r_0'$, and on $\bm{a}^{1:t-1} = \bm{\alpha}^{1:t-1}$ (where we may restrict to $\bm{\alpha}^{1:t-1}$ occurring with nonzero probability), we have:
\begin{align*}
    a_j^t &\sim  \E{\apply^*(p_j,(p_{-j},(r_m)_{m\geq T+1-t}, (x_m)_{m\geq 0}))\mid r_{T+1-t} = r_{0}', \bm{a}^{1:t-1} = \bm{\alpha}^{1:t-1}}\\
    &= \E{\apply^*(p_j,(p_{-j},(r_m)_{m\geq 0}, (x_m)_{m\geq 0}))\mid r_{0} = r_{0}', \bm{a}^{1:t-1} = \bm{\alpha}^{1:t-1},T=t-1}\\
    &= \tau^*_j(\bm{a}^{1:t};(r_0'-\eps)/(1-\eps)).
\end{align*}
We have that $(r_{T+1-t}-\eps)/(1-\eps) \sim U[0,1]$, so we again need only check independence given $r_{T+1-t}$ and $\bm{a}^{1:t-1}$. We have that for $j\neq i$, $\tau^*_j(\bm{a}^{1:t};(r_0'-\eps)/(1-\eps))$ is again deterministic (given $r_{T+1-t}=r_0'$ and $\bm{a}^{1:t-1}$), by \Cref{lemma:eGpB_policy_correl}. Hence, by similar argument to before, we are done.
\end{proof}
\correlFolkThm*
\begin{proof}\label{pf:correlFolkThm}
    To see that feasibility and strict individual rationality are sufficient, we simply apply \Cref{thm:repFolkThm,thm:correl_utils}. Note that correlated $\eps$Grounded$\pi$Bots may implement any policy in a correlated repeated game with screening for which the screened policy is screened whenever the unscreened policy randomises independently (i.e., is random for fixed $q$). Since our proof of \Cref{thm:repFolkThm} used screened and unscreened policies satisfying this constraint, and applying \Cref{thm:correl_utils}, we are done.

    Note that feasibility and individual rationality are also necessary:
    We can see immediately that feasibility is a necessary condition. Moreover, individual rationality is necessary, since otherwise some player $i$ could deviate to a program that simply plays the minimax policy for player $i$ (and all programs would then still halt).
\end{proof}
\subsection{Proofs for Section \ref{sec:uncorrel}}
\haltingUncorrel*
 \begin{proof}\label{pf:haltingUncorrel}
 Note that this is equivalent to a special case of \Cref{thm:correl_halting}, in which all programs pass on the $(x_m)$ unmodified, say, and no program ever otherwise uses the $(x_m)$. Thus, this result follows as an immediate corollary.
 \end{proof}
 \expectedRuntime*
\begin{proof}\label{pf:expectedRuntime}
The proof is almost identical to that of \Cref{thm:expected_runtime_correl}, with $\apply^*$ calls replaced by $\apply$ calls.
\end{proof}
 \uncorrelHaltingExt*
  \begin{proof}\label{pf:uncorrelHaltingExt}
 Note that this is equivalent to a special case of \Cref{thm:correl_halting_ext}, in which all programs pass on the $(x_m)$ unmodified, say, and no program ever otherwise uses the $(x_m)$. Thus, this result follows as an immediate corollary.
 \end{proof}

  \begin{cor}\label{cor:action_sequences}
     For any strategy $s_i$ for player $i$ in game $G$, for all $\delta>0$, for $\eps>0$ (depending on $\delta$) sufficiently small, there is a deterministic sequence of actions $(\alpha_i^t)$ such that $s_i = \sum_{t=1}^\infty \eps (1-\eps)^{t-1}\alpha_i^t$ and for all $t$, we have $\norm{\sum_{k=t}^\infty \eps(1-\eps)^{k-t}\alpha_i^t - s_i}_\infty<\delta.$
 \end{cor}
 \begin{proof}
     If each action has positive support in $s_i$, this follows by applying \Cref{lemma:action_sequences} to a game in which we add players for each possible action of player $i$, and setting the utility of each such player to be 1 if the action corresponding to that player is played and zero otherwise. If some actions do not have positive support, we may simply consider the game where those actions are removed from the action space of player $i$, and then do the same.
 \end{proof}
\uncorrelEquilibria*
\begin{proof}\label{pf:uncorrelEquilibria}
    Let payoffs $v_{1:n}$ satisfy the above conditions with respect to $s_{1:n}$ and $\lambda$. Apply \Cref{cor:action_sequences} to choose $\eps> 0$ small enough and sequences $(\alpha_i^t)$ for each $i$ such that $s_i = \sum_{t=1}^\infty \eps (1-\eps)^{t-1}\alpha_i^t$ and for all $t$, we have $\norm{\sum_{k=t}^\infty \eps(1-\eps)^{k-t}\alpha_i^t - s_i}_\infty<\delta$, for some $\delta>0$ to be determined. Set $\pi_j(\bm{a}^{1:t})$ to be equal to $\alpha_j^{t+1}$ provided $a^{1:t}_k = \alpha^k_{1:t}$ for all $k$. Then, when for $t'\leq t$, $\bm{a}^{1:t'-1} = \bm{\alpha}^{1:t'-1}$, and $a^t_{-i} = \alpha_{-i}^{t'}$ but $a^{t'}_i\neq \alpha^{t'}_i$, set $\pi_j(\bm{a}^{1:t}) = s_{j,i}(a_i^t)$, for some strategy $s_{j,i}$ (to be determined) for player $j$. Otherwise, set $\pi_j$ arbitrarily. We will show that for each $i$, player $i$ obtains utility at most $\max_{q, s_i'}\allowbreak u_i\left((1-q)s_{i}+qs_i',\allowbreak((1-q)s_j+qS_j'\bullet s_i')_{j\neq i}\right) + O(\delta q)$.

    Consider the game between programs $p_{-i}$ and $p_i'$. Let $q$ be the probability that $a_i^{1:T}$ is inconsistent with $\pi_i$, given $T>0$ (where $T$, $\bm{a}^{1:T}$ are as in the definition of the $\eps$Grounded$\pi$Bots). Note that the distribution of $a_i^T$ conditional on $T \geq 1$ is just the same as the marginal distribution of the action of $p_i'$.
    
    Now with probability $1-q$, we have $a_i^{1:T}$ is consistent with $\pi_i$, and hence $a_i^T = \alpha_i^T$. We will now show that conditional on player $i$ not deviating on the current or previous time steps, her action is distributed as $s_i + O(\delta q)$.

    Note that $\bm{a}^{1:t}$ is independent of $T$, conditional on $T\geq t$. Then, let $\delta_t$ be the probability of deviation by time $t$ (i.e., of $a_i^{1:t}$ not being consistent with $\pi_i$), given $T\geq t$. Also set $\delta_0 = 0$. We have that $q= \sum_{t=1}^\infty \prob{T=t\mid T\geq 1}\delta_t$. Meanwhile, 
    \begin{align*}
        \prob{T=t\mid \text{No deviation in } a_i^{1:T}, T\geq 1}&= \prob{\text{No deviation in }a_i^{1:T} \mid T = t}\frac{\prob{T=t\mid T\geq 1}}{\prob{\text{No deviation in }a_i^{1:T} \mid T \geq 1}}\\
        &= \frac{(1-\delta_t)\prob{T=t\mid T\geq 1}}{1-q}.
    \end{align*}
    Thus, given that no deviation occurs in $a_i^{1:T}$, the distribution of $a_i^T$ is:
    \begin{align*}
        \tilde{s}_i &=\sum_{t=1}^\infty \alpha_i^t \prob{T=t \mid  \text{No deviation in } a_i^{1:T}, T\geq 1}\\
        &= \frac{1}{1-q}\sum_{t=1}^\infty \prob{T=t\mid T\geq 1}(1-\delta_t)\alpha_i^t\\
        &= \frac{1}{1-q}\sum_{t=1}^\infty\prob{T=t\mid T\geq 1}\alpha_i^t - \frac{1}{1-q}\sum_{t=1}^\infty\prob{T=t\mid T\geq 1}\delta_t\alpha_i^t\\
        &= \frac{1}{1-q}\sum_{t=1}^\infty\prob{T=t\mid T\geq 1}\alpha^i_t - \frac{1}{1-q}\sum_{t=1}^\infty\prob{T=t\mid T\geq 1}\delta_t(\alpha_i^t-s_i) - \frac{q}{1-q}s_i\\
        &= \frac{s_i}{1-q} - \frac{q}{1-q}s_i - \frac{1}{1-q}\sum_{t=1}^\infty\prob{T=t\mid T\geq 1}\delta_t(\alpha_i^t-s_i)\\
        &= s_i - \frac{1}{1-q}\sum_{t=1}^\infty(\delta_t-\delta_{t-1})\sum_{k\geq t}\prob{T=k\mid T\geq 1}(\alpha^k_i-s_i)\\
        &= s_i - \frac{1}{1-q}\sum_{t=1}^\infty(\delta_t-\delta_{t-1})\prob{T\geq t\mid T\geq 1}(\E{\alpha_i^T\mid T\geq t}-s_i)\\
        &= s_i - \frac{1}{1-q}\sum_{t=1}^\infty(\delta_t-\delta_{t-1})\prob{T\geq t\mid T\geq 1}(\E{\alpha_i^T\mid T\geq t}-s_i).
    \end{align*}
    Now, we have that $(\alpha_i^t)$ by assumption has $\norm{\E{\alpha^i_T\mid T\geq t} - s_i}_\infty \leq \delta$. Thus,
    \begin{align*}
        \norm{\tilde{s}_i-s_i} &\leq \frac{1}{1-q}\sum_{t=1}^\infty(\delta_t-\delta_{t-1})\prob{T\geq t\mid T\geq 1}\delta\\
        &= \frac{1}{1-q}\sum_{t=1}^\infty \delta_t\prob{T=t\mid T\geq 1}\delta\\
        &= \frac{\delta q}{1-\delta}.
    \end{align*}
    Now, write $s_i'$ for the distribution of $a_i^T$ given that some deviation occurs in $a_i^{1:T}$. We have that the overall distribution of $a_i^T$ (conditional on $T\geq 1$) is $(1-q)\tilde{s}_i + q s_i'= (1-q)s_i + q s_i' + O(\delta q)$. This is then the distribution of the action of player $i$.
    
    Meanwhile, each other player $j$ plays $\pi_j()$ when $T=0$ and otherwise observes $\bm{a}^{1:T}$ and plays $\pi_j(\bm{a}^{1:T})$. Now, conditional on no deviation occurring, we have
        \begin{align*}
        \prob{T=t\mid \text{No deviation in } a_i^{1:T}}&= \prob{\text{No deviation in }a_i^{1:T} \mid T = t}\frac{\prob{T=t}}{\prob{\text{No deviation in }a_i^{1:T}}}\\
        &= \frac{(1-\delta_t)\prob{T=t}}{1 - (1-\eps)q}.
    \end{align*}
    and so condition on no deviation, the overall strategy of player $j$ is given by
    \begin{align*}
        \tilde{s}_j &= \frac{1}{1-(1-\eps)q}\sum_{t=0}^\infty\prob{T=t}(1-\delta_t)\alpha_j^{t+1}\\
        &= \frac{s_j}{1-(1-\eps)q} - \frac{1}{1-(1-\eps)q}\sum_{t=0}^\infty \prob{T=t}\delta_t \alpha_j^{t+1}\\
        &= \frac{s_j}{1-(1-\eps)q} -\frac{(1-\eps)qs_j}{1-(1-\eps)q} - \frac{1}{1-(1-\eps)q}\sum_{t=1}^\infty \prob{T=t}\delta_t (\alpha_j^{t+1}-s_j)\\
        &= s_j -\frac{1-\eps}{1-(1-\eps)q}\sum_{t=1}^\infty \prob{T=t\mid t\geq 1}\delta_t(\alpha_j^{t+1}-s_j)\\
        &= s_j -\frac{1-\eps}{1-(1-\eps)q}\sum_{t=1}^\infty (\delta_t-\delta_{t-1})\prob{T\geq t\mid t\geq 1}(\E{\alpha_j^{T+1}\mid T\geq t}-s_j).
    \end{align*}
    Now, we have $\E{\alpha_j^{T+1}\mid T\geq t} = \E{\alpha_j^{T+1}\mid T+1\geq t+1} = \E{\alpha_j^{T}\mid T\geq t+1}$, and so $\norm{\E{\alpha_j^{T+1}\mid T\geq t}-s_j}= \norm{\E{\alpha_j^{T}\mid T\geq t+1}-s_j} \leq \delta$. Hence,
    \begin{align*}
        \norm{\tilde{s}_j-s_j} &\leq \frac{1-\eps}{1-(1-\eps)q}\sum_{t=1}^\infty (\delta_t-\delta_{t-1})\prob{T \geq t\mid T\geq 1}\delta\\
        &= \frac{1-\eps}{1-(1-\eps)q}\sum_{t=1}^\infty\delta_t \prob{T=t\mid T\geq 1}\delta\\
        &= \frac{(1-\eps)q}{1-(1-\eps)q}\delta.
    \end{align*}
    Meanwhile, conditional on a deviation occuring in $a_i^{1:T}$, $a_i^T$ is distributed as $s_i'$ and so player $j$'s action is distributed as $s_{j,i}'\bullet s_i'$. This happens with overall probabilty $(1-\eps)q$. Hence, player $j$'s action is overall distributed as $(1-(1-\eps)q)\tilde{s}_j + (1-\eps)qs_{j,i}\bullet s_i' = (1-(1-\eps)q)s_j + (1-\eps)qs_{j,i}\bullet s_i' + O(\delta q) = (1-q)s_j + qs_{j,i}\bullet s_i' + O(\delta q)$.

    Putting this all together, player $i$ obtains utility $u_i((1-q)\tilde{s}_i+qs_i', ((1-(1-\eps)q)\tilde{s}_j + (1-\eps)qs_{j,i}\bullet s_i')_{j\neq i}) = u_i((1-q)s_i+qs_i', ((1-q)s_j + qs_{j,i}\bullet s_i')_{j\neq i}) $, since the utility function is multilinear and hence Lipschitz. This is then at most $\max_{q, s_i'}\allowbreak u_i\left((1-q)s_{i}+qs_i',\allowbreak((1-q)s_j+qs_{j,i}\bullet s_i')_{j\neq i}\right)$, up to $O(\delta q)$ terms.

    Now, by taking $\eps$ sufficiently small, we may ensure that the $O(\delta q)$ terms are at most $\frac{\lambda}{2}q$. If no player deviates, each player $i$ just receives expected utility $v_i$.

    Then, choose $s_{j,i}$, for each $j$ and $i$, so as to minimise $\max_{q, s_i'}u_i((1-q)s_{i}+qs_i',(1-q)s_{-i}+qs_{-i,i}\bullet s_i')$ (possible by continuity and compactness). We have that player $i$ by deviating can gain additional utility at most
    \begin{equation*}
        \max_{q, s_i'}u_i((1-q)s_{i}+qs_i',(1-q)s_{-i}+qs_{-i,i}\bullet s_i') - v_i + \frac{\lambda}{2}q \leq -\frac{\lambda}{2}q < 0.
    \end{equation*}
    Hence, this is an equilibrium.
\end{proof}

\uncorrelEquilibriaNecessary*
\begin{proof}\label{pf:uncorrelEquilibriaNecessary}
    That $v_{1:n}$ must be feasible without correlation is immediate. Suppose, for each $j$, $p_j$ is an $\eps$Grounded$\pi$Bot with discount factor $\eps_j$ and policy $\pi_j$, and that $p_1,\ldots p_n$ form a program equilibrium with strategies $s_{1:n}$ and payoffs $v_{1:n}$. Let $\eps = \min_j \eps_j$, and $T((r_m)) = \min\{t:r_t<\eps\}$. Then, construct program $p_i'$ as follows:
    \begin{itemize}
        \item Check if $r_{T+1}< q$. If so, use $r_{T+1}$ to randomise according to $s_i'$.
        \item Otherwise, output $\apply(p_i,(p_{-i},(r_m))$.
    \end{itemize}
    \Cref{thm:halting_uncorrel} gives us that $p_{-i}$ and $p_i'$ halt against each other almost surely (since $p_i'$ at worst just runs $p_i$, and the structure of the tree of simulations made by the $\eps$Grounded$\pi$Bots does not depend on the output of programs).

    Now, let $S_j'(a)$ be the overall strategy of $p_j$ if $p_i$ were instead simply replaced with a program that always output action $a$.

    Note that none of the $p_j$ ever makes an apply call (including within other $\apply$ calls) with input sequence $(r_m)_{m\geq k}$ with $k>T$. Then, conditional on $r_{T+1} \geq q$, we have that $p_i'$s output will always be compatible with $\pi_i$ (and $p_i$). So, since the distribution $r_{T+1}$ is independent of $T$ and of $r_{1:T}$, we have that conditional on $r_{T+1}\geq q$, the strategies of all programs is given by $s_{1:n}$. Meanwhile, conditional on $r_{T+1}< q$, given that $T\geq 1$, all $\apply$ calls referencing $p_i'$ made by the other programs will have output $a$, for some fixed action $a$, with probability proportional to $s_i(a)$. Thus, program $j$ for $j\neq i$ must then have overall strategy $S_j'(a)$.

    Thus, overall, $p_i'$ has strategy $(1-q)s_i+qs_i'$ and for each $j\neq i$, we have that $p_j$ has overall strategy $(1-q(1-\eps))s_j+ q(1-\eps)\sum_a s_i'(a)S_j'(a)$, which is in turn equal to $(1-q)s_j + q \tilde{S}_j\bullet s_i'$, where $\tilde{S}_j(a) \defeq \eps s_j + (1-\eps) S_j'(a)$. 

    Since this is an equilibrium, we therefore obtain that, for some $\tilde{S}_j:A_i \rightarrow\Delta(A_j)$, for all $s_i'$ and $q$,
    \begin{align*}
        v_i &\geq u_i\left((1-q)s_i+qs_i', \left((1-q(1-\eps))s_j+ q \tilde{S}_j\bullet s_i'\right)_{j\neq i}\right).
    \end{align*}
    Maximising over $q$ and $s_i'$ for the deviator, and then minimising over the $\tilde{S}_j$, the result immediately follows.
\end{proof}

\subsection{Proofs for Section \ref{sec:simulationist}}
\outputIndepHalting*
\begin{proof}\label{pf:outputIndepHalting}
    Suppose that $p_1,\ldots p_n$ are $\eps$Grounded$\pi$Bots and we replace $p_i'$ with a program that halts on a given input whenever $p_i$ halts. Consider the resulting call tree, truncating whenever we reach a node at which $p_i'$ is called. We have that the call tree has at most the depth of the original call tree, and that all of the nodes corresponding to calls to $p_i'$ halt, since the same call to $p_i$ occurs in the original call tree.
\end{proof}

\simulationistNegativeResultGame*
\begin{proof}\label{pf:simulationistNegativeResultGame}
We will proceed by constructing $T_1,\ldots T_n$ and a policy profile for the new game. We will show that this policy profile grants the same payoffs to players as in the equilibrium of the program game. Moreover, for each possible deviation of the given form in the new game, we will show that there exists a corresponding deviation of the program game in which all players receive the same utilities. From this, we deduce that no such deviation can be profitable in the new game.

We begin by defining $T_1, \ldots, T_n$ as functions of the random sequence $(r_m)$. For each $i$, let $T_i((r_m))$ be the maximum depth of nested $\apply$ calls made by $p_i$ on input $(p_{-i},(r_m))$. I.e., if $p_i$ makes no $\apply$ calls on this input, $T_i=0$, and otherwise $T_i$ is one more than the maximum depth of $\apply$ calls made by $p_i$'s own $\apply$ calls. In essence, $T_i+1$ can be thought of as the `time step' of player $i$ in the program game.

We now work towards defining the policies in the new game (on relevant strategy histories). We define these based on the action distributions of programs in the program game against a particular class of programs. For $t= 1, 2, \ldots$, and each $i$, let $c_i^t$ be the strategy for player $i$ corresponding to the distribution of $\apply(p_i,(p_{-i},(r_m)))$ conditional on $T_i+1 = t$. Then, given sequence $d_i$ and $t_0$, define program $p_i'(t_0,d_i)$ as follows: Given input $(p_{-i},(r_m))$, determine $T_i((r_m))$, by running $p_i$ on this input and counting the maximum depth of $\apply$ calls. If $T_i < t_0$, output $\apply(p_i,(p_{-i},(r_m)))$. Otherwise, use $r_0$ to randomise according to $d_i$ (in some arbitrary fixed way). Finally, if the input is $(q_{-i},(r_m))$ with $q_{-i}\neq p_{-i}$, simply output $\apply((p_i,(q_{-i},(r_m)))$.

Note that $p_i'$ halts on every input on which $p_i$ halts, and so by output independent halting we have that programs $p_{-i}$, $p_i'$ halt against each other (for any choice of $i$ and $d_i$). Now, for each $t_0$ and $d_i$ and associated $p_i'$, for each $j \neq i$, let $s_j(i,t_0,d_i) = \E{\apply(p_j,(p_{-ij},p_i', (r_m)))\mid T_j > t_0}$.

We are now ready to define policies $\tau_{1:n}$ for the players in the new game (as functions of the observed strategy history). First, set $\tau_i(\omega) = c_i^1$ and $\tau_i(c_{1:n}^{1:t}) = c_i^{t+1}$ for all $i$ and $t$. Then, if strategy history $\sigma_{1:n}^{1:t}$ is such that $\sigma_{1:n}^{1:t_0-1} = c_{1:n}^{1:t_0-1}$ for some $t_0$, and $\sigma_{-j}^{t_0} = c_{-j}^{t_0}$ and $\sigma_j^{t_0}\neq c_j^{t_0}$, let $\tau_i(\sigma_{1:n}^{1:t}) = s_i(j,t_0,\sigma_j^{t_0})$. Elsewhere, define $\tau_i$ arbitrarily.

Let's now show that $\tau_{1:n}$ satisfy the desired properties. First, we have that when all players play according to $\tau_{1:n}$, each player $i$ plays $c_i^t$ conditional on $T_i+1=t$. Thus, the overall expected utility obtained by player $i$ is $u_i(\E{c_1^{T_1+1}},\ldots, \E{c_n^{T_n+1}})$, which is in turn equal to $u_i(\E{\E{\apply(p_1(p_{-1},(r_m)))\mid T_1+1}},\ldots, \E{\E{\apply(p_n(p_{-n},(r_m)))\mid T_n+1}}) = u_i(\E{\apply(p_1(p_{-1},(r_m)))},\ldots, \E{\apply(p_n(p_{-n},(r_m)))}) = v_i$, the same as in the program game, as required. Now, suppose that player $i$ deviates to start playing $d_i$ for $T_i\geq t_0$. We then have that for each $j\neq i$, for $T_j = t\leq t_0+1$, player j plays $c_j^t$. Otherwise, player $j$ must play $s_j(i,t_0,d_i)$. 

Meanwhile, suppose that in the program game, player $i$ deviates to program $p_i'(t_0,d_i)$. Then, conditional on $T_j = t \leq t_0$, we know that $p_j$ on input $(p_{-j},(r_m))$ makes $\apply$ calls of maximum depth $t$. In that case, the call tree cannot contain anywhere any calls to $p_i$ with input sequence $(r_m')$ with $T_i((r_m')) \geq t$. Meanwhile, by the definition of simulationist, and the fact $p_1,\ldots,p_n$ are all simulationist, all calls in the tree must be to simulationist programs. We will now argue that for each call in the call tree, the output must be the same if any number of instances of $p_i$ are replaced by $p_i'$. Suppose this were not the case. Then let $\apply(p,(q_{1:n-1},(r_m')))$ be a maximal depth call for which this is not the case (where the calls made directly by the top-level program have depth 1, and so on). If $p$ is $p_i$, we have that $T_i((r_m')) < t \leq t_0$, and so if $q_{1:n-1}=p_{-i}$, we have that $p_i'$ would give the same output as $p_i$. Otherwise, $p_i'$ would run the same apply calls as $p_i$ given input $(q'_{1:n-1},(r_m'))$, where each $q_i'$ is either $q_i$, or if $q_i=p_i$, $q_i'$ might be $p_i'$. Since $p_i$ is simulationist, the first such call must be the same as if the input programs were $q_i$ up to substitutions (since $q_i'$ cannot be distinguished from $q_i$ except via simulations), and so by assumption must give the same output as it would in that case. Given that this holds for all calls made by $p_i$ on this input so far, it must then also hold for the next one. So by induction the same direct calls (up to substitutions) are made, and give the same output, so $p_i$ must give the same output with $q_i'$ as with $q_i$. Hence, $\apply(p_i',(q_{1:n-1}',(r_m'))) = \apply(p_i,(q_{1:n-1},(r_m')))$, a contradiction. Thus, we must have $p \neq p_i$. But then by the same induction argument as before, $p$ must make the same calls (up to substitutions), which must give the same outputs. Hence, each call in the call tree gives the same output when any number of instances of $p_i$ is replaced by $p_i'$. From this, we may deduce, by similar argument, that $p_j$ itself gives the same output when $p_i$ is replaced by $p_i'$, conditional on $T_j=t\leq t_0$.

So, when player $i$ deviates to $p_i'$, conditional on $T_j = t \leq t_0$, we have that $p_j$'s action is distributed as $c_j^t$. Meanwhile, conditional on $T_j>t_0$, we know $p_j$'s action is distributed as $s_j(i,t_0,d_i)$, by construction.

Thus, the action distribution of each player $j$ in the program game, conditional on $T_j$, is the same as in the new game. Since the $T_j$ have the same distribution in both, we may conclude that all players receive the same utilities in both. In particular, since $p_i'$ cannot be a profitable deviation in the program game, $\tau_i'$ cannot be a profitable deviation in the new game. We are then done.

\end{proof}
\simulationistNegativeResult*
\begin{proof}\label{pf:simulationistNegativeResult}
We will apply \Cref{prop:simulationistNegativeResultGame} and its proof. Let $T_i$, $c_i^t$, $\tau_i$ and $s_i(j,t_0,d_j)$ be as in the proof of \Cref{prop:simulationistNegativeResultGame} for each $i$, $t$, $t_0$, $j$, $d_j$. Let $\delta_i^t = \prob{T_i+1=t}$. Then we have from \Cref{prop:simulationistNegativeResultGame} that if player $i$ (in the game from \Cref{prop:simulationistNegativeResultGame}) deviates to play $d_i$ when $T_i+1\geq t_0$, this cannot be profitable. Meanwhile, this results in overall strategy $\sum_{t=1}^{t_0-1}\delta_i^t c_i^t +d_i \sum_{t\geq t_0} \delta_i^t$ for player $i$. For $j\neq i$, player $j$ then plays according to overall strategy $\sum_{t=1}^{t_0}\delta_j^t c_j^t + s_j(i,t_0, d_i) \sum_{t>t_0}\delta_j^t$. Thus, we must have
\begin{align*}
    v_i &\geq u_i\left(\sum_{t=1}^{t_0-1}\delta_i^t c_i^t +d_i \sum_{t\geq t_0} \delta_i^t, \left(\sum_{t=1}^{t_0}\delta_j^t c_j^t + s_j(i,t_0, d_i) \sum_{t>t_0}\delta_j^t\right)_{j\neq i}\right)\\
    &\geq \min_{s_{-i}}u_i\left(\sum_{t=1}^{t_0-1}\delta_i^t c_i^t +d_i \sum_{t\geq t_0} \delta_i^t, \left(\sum_{t=1}^{t_0}\delta_j^t c_j^t + s_j \sum_{t>t_0}\delta_j^t\right)_{j\neq i}\right)
\end{align*}
for all $t_0$ and $d_i$. Taking the maximum over $t_0$ and $d_i$ gives the result.

\end{proof}
\noRuleOutPareto*
\begin{proof}\label{pf:noRuleOutPareto}
    Let $(c_1,c_2)$ be a strictly Pareto optimal in the limit of small deviations, strictly individually rational policy profile.  WLOG assume $u_1(c_1,c_2) = u_2(c_1,c_2) = 0$ (by translating utilities if necessary)/
    
    Consider the policy of playing according to strategies $(c_1,c_2)$, and if player $i$ deviates to policy $d_i$ starting at time step $t_0$, punishing with whatever strategy minimises the player's expected utility given their overall policy (assuming that they will continue to play according to $d_i$). I.e., choose $s_{-i}$ so as to minimise $u_i\left(\sum_{t=1}^{t_0-1}\delta)i^t c_i^t + d_i \sum_{t\geq t_0}\delta_i^t, \sum_{t=1}^{t_0}\delta_{-i}^tc_{-i}^t + s_{-i}\sum_{t>t_0}\delta_{-i}^t\right)$, and punish using $s_{-i}$, where the sequences $(\delta_{1:2}^t)$ are to be determined. We will show that for appropriate choice of $T_1,T_2$ (or equivalently $(\delta_1^t),(\delta_2^t)$, this is an equilibrium of the game described in \Cref{prop:simulationistNegativeResultGame}).

    Now, given that player 1 deviates with probability $p$ to strategy $d_1$, let $q_2(d_1,a_2,p)$ be the minimum probability of detection needed to effectively punish player 1 using action $a_2$ -- i.e., $q_2$ is minimal such that $u_1((1-p)c_1+pd_1, (1-q_2)c_2+q_2 a_2)\leq 0$.

    Then, let $q_2(p) = \max_{d_1}\min_{a_2} q_2(d_1,a_2,p)$. Note that we need player 2 to detect player 1 with probability at least $q_2(p)$ when player 1 deviates with probability $p$, in order to enforce $(c_1,c_2)$. Define $q_1(p)$ similarly (where player 2 deviates and player 1 punishes).

    Note that since $(c_1,c_2)$ is strictly individually rational, we must have that $q_1$ and $q_2$ are well defined and $q_i(p)<1$ for all $p$. (This is because if a player punishes with probability 1, they can ensure the other player gets utility \emph{strictly} less than 0, and so by continuity of the utility in the punishment probability, there must be some punishment probability less than one which allows for utility at most zero for the deviating player.)

    It now remains to show that we can construct $T_1$ and $T_2$ such that for all $t$, $\prob{T_1\geq t+1} \geq q_1(\prob{T_2\geq t})$, and vice versa. Now, if either player $i$ has no profitable deviations (i.e., no actions that perform better than $c_i$ against $c_{-i}$), then we may simply set $T_{-i}=0$ with probability 1, and $T_i = 1$ with probability 1. Otherwise, consider setting $\prob{T_1\geq t} = (1-\eps)^t$ for some small $\eps>0$, to be determined. Then for each $t\geq 1$, set $\prob{T_2\geq t} = q_2((1-\eps)^{t-1})$. This is well defined, and by definition we have that $\prob{T_2\geq t+1} \geq q_2(\prob{T_1\geq t})$, for all $t$.

   For the other direction, we have for $t\geq 1$, 
   \begin{align*}
       q_1(\prob{T_2\geq t}) &= q_1(q_2((1-\eps)^{t-1}))\\
       &= (1-\eps)^{t-1}\frac{q_1(q_2((1-\eps)^{t-1}))}{(1-\eps)^{t-1}}\\
       &= \prob{T_1\geq t+1}(1-\eps)^{-2}\frac{q_1(q_2((1-\eps)^{t-1}))}{(1-\eps)^{t-1}}\\
       &\leq \prob{T_1\geq t+1}(1-\eps)^{-2}\sup_{p\in(0,1]}\frac{q_1(q_2(p))}{p}.
   \end{align*}
   Thus, it suffices to show that $\sup_{p\in(0,1]}\frac{q_1(q_2(p))}{p} < 1$, so that for sufficiently small $\eps$, we have $(1-\eps)^{-2}\sup_{p\in(0,1]}\frac{q_1(q_2(p))}{p}<1$. We will proceed by showing that $\frac{q_1(q_2(p))}{p} < 1$ for all $p>0$, and that the limit of $\frac{q_1(q_2(p))}{p}$ as $p\rightarrow 0$ exists and is less than 1. We will then use compactness of the interval $[0,1]$ to obtain our conclusion.

   First, we will show that $q_1(q_2(p)) < p$ (and similarly for $q_2(q_1(p))$), for $p>0$. To see this, consider any $p$, and possible deviation $d_2$ for player 2. Suppose player 2 deviates to $d_2$ with probability $q_2(p)$. Then if player 1 detects this deviation with probability $p$, either for some $a_1$ this deviation and punishment is strictly bad for player 2 (showing $q_1(q_2(p))<p$), or for every possible $a_1$ it is at worst neutral for player 2. But then by strict Pareto optimality, it must be negative for player 1 for every possible $a_1$ (assuming the deviation is not just to $c_2$). However, this means that for any possible strategy that \emph{player 1} could deviate to with probability $p$, \emph{player 2 punishing} with $d_2$ with probability $q_2(p)$ would result in a strictly negative payoff. This contradicts the definition of $q_2$ (we could make it smaller, if this were the case), and so we are done. The case for $q_2(q_1(p))$ is symmetric.

   Note that $q_1$ and $q_2$ are continuous, (non-strictly) increasing in $p$, and $q_1(0)=q_2(0) = 0$. Now, in the limit of small $p$, the utility for player $i$ of player $i$ deviating to $d_i$ with probability $p$ and player $-i$ punishing with some $s_{-i}$ with probability $q_{-i}(p)$ is $p u_i(c_{-i},d_i)+ q_{-i}(p)u_i(s_{-i},c_i) + o(p)$. This is then at most zero, and zero in the worst case, which as $p\rightarrow0$ tends towards $d_i = d_i^*$ where $d_i^*$ is chosen to maximise $u_i(c_{-i},d_i^*)$, and $s_{-i} = s_{-i}^*$, chosen to minimise $u_i(s_{-i},c_i)$. Recall that we assumed both players have positive deviations, and hence $u_i(c_{-i},d_i^*)>0$ and $u_{-i}(s_{-i}^*,c_i)>0$ and so by strict pareto optimality, $u_{-i}(s_{-i}^*,c_i)<0$. Thus, in the limit of small $p$, we must have $p u_i(c_{-i},d^*_i)+ q_{-i}(p)u_i(s_{-i}^*,c_i) = o(p)$, and $\frac{q_{-i}(p)}{p} \rightarrow \frac{u_i(c_{-i},d_i^*)}{-u_i(s_{-i}^*,c_i)}$, which is well-defined and positive. Thus, as $p\rightarrow 0$, we have that 
   \begin{align*}
       \frac{q_i(q_{-i}(p))}{p} &= \frac{q_i(q_{-i}(p))}{q_{-i}(p)}\frac{q_{-i}(p)}{p}\\
       &\rightarrow 
       \frac{\max_{d_1}u_1(d_1,c_2)}{-\min_{s_2}u_1(c_1,s_2)} \frac{\max_{d_2}u_2(c_1,d_2)}{-\min_{s_1}u_2(s_1,c_2)}\\
       &\leq \frac{u_1(d^*_1,c_2)}{-u_1(c_1,d_2^*)} \frac{u_2(c_1,d^*_2)}{-u_2(d_1^*,c_2)}
   \end{align*}
   
   It remains to show that this is less than 1. Let $\rho = \frac{u_1(d^*_1,c_2)}{-u_1(c_1,d^*_2)}$. Consider $s_1 = \frac{1}{1+\rho}d_1^*+\frac{\rho}{1+\rho}c_1$, and $s_2 = \frac{\rho}{1+\rho}d_2^*+\frac{1}{1+\rho}c_2$. By strict Pareto optimality in the limit of small deviations, we know that $u_i(c_1,s_2)+u_i(s_1,c_2)$ is negative for at least one player. Now, by construction,
   \begin{equation*}
       u_1(c_1,s_2)+u_1(s_1,c_2) = \frac{\rho}{1+\rho} u_1(c_1,d_2^*) + \frac{1}{1+\rho}u_1(d_1^*,c_2)=0.
   \end{equation*}
   Thus, it must be negative for player 2:
   \begin{align*}
       0 &> u_2(c_1,s_2)+u_2(s_1,c_2)\\
       &= \frac{\rho}{1+\rho} u_2(c_1,d_2^*) + \frac{1}{1+\rho}u_2(d_1^*,c_2)\\
       &\propto \frac{u_1(d^*_1,c_2)}{-u_1(c_1,d^*_2)}u_2(c_1,d^*_2)+u_2(d^*_1,c_2).
   \end{align*}
   So $\frac{u_1(d^*_1,c_2)}{-u_1(c_1,d^*_2)}u_2(c_1,d^*_2) < -u_2(d^*_1,c_2)$, and so
   \begin{equation*}
       \lim_{p\rightarrow 0}\frac{q_i(q_{-i}(p))}{p}\leq \frac{u_1(d_1,c_2)}{-u_1(c_1,d_2)}\frac{u_2(c_1,d_2)}{-u_2(d_1,c_2)} < 1.
   \end{equation*} 
   
   Now, since since each $q_i$ must be continuous, we have that the function defined by $\frac{q_i(q_{-i}(p)))}{p}$ on $(0,1]$, and the limit of that function at $0$, is then continuous in $p$, and so by compactness of the interval $[0,1]$, we have that this function obtains its supremum on $[0,1]$. Its supremum must then be less than 1. It follows that the supremum of $\frac{q_i(q_{-i}(p)))}{p}<1$ on the interval $(0,1]$ must be less than 1, and we are done.

\end{proof}

\section{Additional details for examples}
\subsection{\Cref{eg:pirates1}}\label{app:correl_pirates}
    Let us explicitly describe an $\eps$Grounded$\pi$Bot (\Cref{algorithm:pi-bot-correl}) equilibrium in which each program plays $C$ deterministically at equilibrium. Define the policy for player $i$ as $\pi_i(\bm{a}^{1:t}) = C$ if $a_j^{t'}= C$ for all $j$ and $t'$, and $\pi_i(\bm{a}^{1:t}) = L$ otherwise.

    Now, suppose that player $i$ deviates to a program $p_i'$, which, against the other programs, sometimes plays some action other than $C$. Let $T+1$ be the time step as defined in the $\eps$Grounded$\pi$Bot, and $X$ be the time of first deviation for player $i$. I.e., $X\geq 1$ is minimal such that $\apply^*(p_i,(p_{-i},(r_m)_{m\geq T+1-X},(x_m)_{m\geq 0})) \neq C$, and $X=\infty$ if said minimum is not well defined. We have that if $X>T+1$, then all players must play $C$, and so player $i$ receives utility 10. Meanwhile, if $X\leq T$ we have that other players play $L$, and so player $i$ receives utility at most 9. Finally, if $X=T+1$, other players must play $C$, whereas player $i$ may play $D$, for a a maximum utility of 14. Note that the observed $a_i^t$s are correlated between players $-i$, and the actions of players $-i$ are correlated with one another. Thus, we have that the expected utility of player $i$ is at most $10\prob{X>T+1} + 14\prob{X = T+1} + 9\prob{X<T+1}$. Thus, the gain in utility from deviating is at most $4\prob{X=T+1} - \prob{X<T+1}$. Now, note that the distribution of a player's action at time $t$ (i.e., with input sequence $(r_m)_{m\geq T+1-t}$) is independent of $T$, conditional on $T+1\geq t$. Thus, $\prob{X=t\mid T+1 = t'}$ is the same for all $t'\geq t$. In particular
    
    \begin{align*}
    \prob{X<T+1} &= \sum_{t}\prob{X=t\cap T+1 > t}\\
    &= \sum_{t}\prob{X=t \mid T+1>t}\prob{T+1>t}\\
    &= \sum_{t}\prob{X=t \mid T+1=t}\prob{T+1>t}\\
    &= \sum_{t}\prob{X=t \mid T+1=t}\prob{T+1=t}\frac{1-\eps}{\eps}\\
    &= \frac{1-\eps}{\eps}\prob{X=T+1}.
    \end{align*}
    
    Hence we have that the gain from deviating is at most $4\prob{X=T+1} - \prob{X<T+1} = 4\prob{X=T+1} - \frac{1-\eps}{\eps}\prob{X=T+1}$. For $\eps\leq \frac{1}{5}$, this is at most 0.
\subsection{\Cref{eg:pirates_uncorr_pi_bots}}\label{app:pirates_uncorr_pi_bots}
We show that we can restrict attention to $s_2 = s_3 = L$ or $s_2=s_3 = D$ in the following:
\begin{equation*}
    \min_{s_2,s_3} u_1((1-q)C+qD, (1-q)C + qs_2, (1-q)C + q s_3)
\end{equation*}

Note that for all $a$, we have $u_1(a,D,a_3) = u_1(a,a_2,D) = u_1(a,D,D)$ for all $a_2, a_3$ and so (considering the case where both players 2 and 3 play $L$, and the case where at least one of them plays $D$):
\begin{align*}
    u_1(a, xD+(1-x)L, yD + (1-y)L) &\geq (1-x)(1-y)u_1(a,L,L) + (x+y-xy)u_1(a,D,D)\\
    &= (1-x)(1-y)u_1(a,L,L) + (1-(1-x)(1-y))u_1(a,D,D)
\end{align*}

and

\begin{align*}
    u_1(a,C,yD+(1-y)L) &= yu_1(a,C,D) + (1-y)u_1(a,C,C)\\
    &\geq (x+y-xy)u_1(a,C,D) + (1-x)(1-y)u_1(a,C,C)\\
    &=  (x+y-xy)u_1(a,C,D) + (1-x)(1-y)u_1(a,C,L)\\
    &= (1-x)(1-y)u_1(a,C,L) + (1-(1-x)(1-y))u_1(a,C,D)
\end{align*}

and similar for $u_1(a,xD+(1-x)L,C)$. Thus, $s_2$ and $s_3$ being mixtures of $D$ and $L$ gives player 1 at least as much utility as a convex combination of both being pure $D$ and both being pure $L$ (since it is at least as much as the same convex combination in all possible cases in terms of whether the probability $q$ or $1-q$ event happens for each player). Thus, we need only consider $s_2=s_3=L$ and $s_2=s_3=D$.

Utility under each of these possibilities:
We have
\begin{align*}
    &u_1((1-q)C+qD,(1-q)C+qL,(1-q)C+qL)\\
    &= q^2u_1((1-q)C+qD,L,L)+ (1-q^2)(1-q)u_1(C,C,C) + (1-q^2)q u_1(D,C,C)\\
    &= 9q^2 + 10(1-q)(1-q^2) + 14q(1-q^2)\\
    &= 10 +4q -q^2 -4q^3,
\end{align*}
which is equal to $10\frac{3}{4}$ when $q=\frac{3}{4}$. Meanwhile

\begin{align*}
    &u_1((1-q)C+qD,(1-q)C+qD,(1-q)C+qD)\\
    &= qu_1(D,(1-q)C+qD,(1-q)C+qD) + (1-q)^3u_1(C,C,C)\\
    &= 14q + 10(1-q)^3,
\end{align*}
which is equal to $10\frac{21}{32}$ when $q=\frac{3}{4}$.
\subsection{\Cref{eg:trust}}\label{app:mixed}

\begin{table}[ht]
\centering
\setlength{\extrarowheight}{3pt}
\begin{tabular}{cc|c|c|}
  & \multicolumn{1}{c}{} & \multicolumn{2}{c}{Player $2$ }\\
  & \multicolumn{1}{c}{} & \multicolumn{1}{c}{$G$}  & \multicolumn{1}{c}{$C$}\\
  \cline{3-4}
  \multirow{2}*{Player 1}  & $K$ & $3,0$ & $3,0$\\
  \cline{3-4}
  & $S$ & $2,4$ & $4,2$\\
  \cline{3-4}
\end{tabular}
\caption{The trust game from \Cref{eg:trust}}
\label{tab:TrustMixed1}
\end{table}

Consider the following trust game (with payoffs shown in \Cref{tab:TrustMixed1}): Player 1 receives \$3, which he may Keep ($K$) or Send to the other player ($S$). Money sent to the other player is doubled. Player 2 may then either be Greedy ($G$) and take $\$4$ and send only $\$2$ back, or she may be Charitable ($C$) and take only $\$2$, giving $\$4$ to player 1.

First, we will show that there exists an uncorrelated $\eps$Grounded$\pi$Bot equilibrium in which player 1 plays $S$ and player 2 randomises at an approximately $1:1$ ratio between $G$ and $C$, for approximately $\$3$ in expectation for both players (with player 1 receiving slightly more). Let both programs have a common value of $\eps$. Then, define the policies as follows: $\pi_2(a_1^{1:t},a_2^{1:t}) = C$ if $t+1$ is odd, and $G$ if $t+1$ is even. $\pi_1(a_1^{1:t},a_2^{1:t}) = K$ if $a_2^t = G$ and $t$ is odd, and $S$ otherwise.

Now, by the proof of \Cref{prop:uncorrelEquilibria}, if player $2$ deviates to another program (that plays some other strategy $s_2'$ with overall probability $q$ for some $q$), she may obtain utility at most $\max_{q, s_2'}u_2((1-q)s_{2}+qs_2',(1-q)s_1+qS_{1}'\bullet s_2')+O(\delta q)$, where $\delta\rightarrow 0$ as $\eps\rightarrow 0$, where $s_1$ and $s_2$ are the strategies played by the $\eps$Grounded$\pi$Bots against each other, $S_1'(C) = S$ and $S_1'(G) = K$. That is, $s_1 = S$ and $s_2 = \frac{1}{2-\eps}C + \frac{1-\eps}{2-\eps}G$. Now, player 2 has no possible reason to deviate to play more $C$. Meanwhile, $u_2(s_1,s_2) = 2+\frac{2-2\eps}{2-\eps}$. Then, if player 2 deviates to play more $G$ (for effective strategy $(1-q)s_2+qG$), she then receives additional utility at most $u_2((1-q)S+qK,(1-q)s_2 + qG)  - u_2(s_2,s_1) + O(\delta q)$, which we may verify is equal to $-q\left(2+q +O(\eps+\delta)\right)$.
For $\eps$ sufficiently small, this is negative for all $q>0$. Hence, this is an equilibrium.

Meanwhile, there is no way of enforcing this equilibrium with \citet{oesterheld2019robust}'s $\eps$Grounded$\pi$Bots. This is because those programs can only implement myopic policies. A myopic policy cannot punish deviation to play $G$ more, without also responding in the same way to $G$ when played at equilibrium. Thus, equilibria with player 2 mixing and player 1 playing pure $S$ cannot be implemented by the 2019 $\eps$Grounded$\pi$Bots.

\end{document}